\documentclass{ws-ijbc}
\usepackage{graphicx}
\usepackage{epstopdf}
\usepackage{amssymb}  
\usepackage{listings}             
\begin{document}

\catchline{}{}{}{}{} 

\markboth{
		N.V.~Kuznetsov,
		D.G.~Arseniev,
		M.V.~Blagov, 
		M.Y.~Lobachev,
		Z.~Wei,
		M.V.~Yuldashev,
		R.V.~Yuldashev
		}{The Gardner problem and cycle slipping bifurcation for type 2 phase-locked loops}

\title{The Gardner problem and cycle slipping bifurcation for type 2 phase-locked loops}

\author{Nikolay V. Kuznetsov}
\address{
	Faculty of Mathematics and Mechanics,
	Saint Petersburg State University, Russia\\
	Faculty of Information Technology,
	University of Jyv\"{a}skyl\"{a}, Finland\\
	Institute for Problems in Mechanical Engineering of the Russian Academy of Sciences, Russia\\
	nkuznetsov239@gmail.com}

\author{Dmitry G. Arseniev}
\address{Peter the Great Saint Petersburg Polytechnic University, Russia \\
Saint Petersburg State University, Russia}

\author{Mikhail V. Blagov}
\address{Faculty of Mathematics and Mechanics,
		Saint Petersburg State University, Russia\\
		Faculty of Information Technology,
		University of Jyv\"{a}skyl\"{a}, Finland}
	
\author{Mikhail Y. Lobachev}
\address{Faculty of Mathematics and Mechanics,
		Saint Petersburg State University, Russia\\
		Industrial Management Department, LUT University, Finland}
	
\author{Zhouchao Wei}
	\address{School of Mathematics and Physics,
		China University of Geosciences, China,}

\author{Marat V. Yuldashev, Renat V. Yuldashev}
\address{Faculty of Mathematics and Mechanics,
	Saint Petersburg State University, Russia}

\maketitle

\begin{history}
\received{(to be inserted by publisher)}
\end{history}

\begin{abstract}
	In the present work, a second-order type 2 PLL with a piecewise-linear phase detector characteristic is analysed. 
	An exact solution to the Gardner problem on the lock-in range is obtained for the considered model.
	The solution is based on a study of cycle slipping bifurcation and improves  well-known engineering estimates.
\end{abstract}


\section{Introduction}

Phase-locked loops (PLLs) are nonlinear control systems which are designed to synchronize a voltage-controlled oscillator (VCO) signal with a reference one.
PLLs have many applications in energy and robotic systems, satellite navigation, wireless and optical communications, cyber-physical systems 
\cite{DuS-2010-communication, Karimi-Ghartemani-2014, rosenkranz2016receiver, BestKLYY-2016, KaplanH-2017-GPS,  KuznetsovVSYY-2020, ZelenskyGVSKC-2021, ZelenskySAAIV-2021, KuznetsovKBTYY-2022-GN}.
Analog PLLs can be described by systems of nonlinear differential equations with periodic right-hand sides, which are also known as pendulum-like systems.
In 1933, F.~Tricomi was the first, who conducted nonlinear analysis \cite{Tricomi-1933} of the systems which are equivalent to the second-order PLLs with lag filters (see, e.g., \cite{Gardner-2005-book}).
It was proven that the global stability of those systems is determined by separatrices of a saddle, which correspond to a heteroclinic bifurcation in the system.
Further, bifurcations of the second-order PLLs with lead-lag filters and different nonlinear characteristics of phase detectors were studied in \cite{AndronovVKh-1937, Kapranov-1956, Belyustina-1959, Gubar-1961, Shakhtarin-1969}.

PLL systems with lag and lead-lag loop filters can be classified as type 1 PLLs, because transfer functions of such filters do not have poles at the origin.
In engineering practice, so-called type 2 PLLs, that have loop filters with exactly one pole at the origin, are most often used nowadays \cite{Gardner-2005-book}.
The second-order type 2 analog PLLs are always globally stable (see, e.g., \cite{KuznetsovLYY-2021-TCASII}), i.e., these PLLs acquire lock for any reference frequency.
However, synchronization in the systems may take long time. 
In order to reduce the long acquisition time, the lock-in concept has been introduced.
According to the concept, the locked PLL re-acquires a locked state without cycle slipping after an abrupt change of the reference frequency. 
The problem of estimation of the reference frequencies where the concept is held was posed by F.~Gardner in his monograph \cite{Gardner-2005-book}.
A rigorous approach to the Gardner problem and analytical estimates of the lock-in range were suggested in \cite{KuznetsovLYY-2015-IFAC-Ranges, KuznetsovLYY-2019-DAN, KuznetsovLYY-2021-TCASII, KuznetsovMYY-2021-TCAS, KuznetsovLYYVS-2021-DAN}.

The system where such abrupt reference frequency change occurs can be considered as a switching system. 
The Gardner problem requires to study cycle slipping bifurcation of the system when a trajectory, starting from an equilibrium of the system before the switch tends to an equilibrium of the system after the switch.
This task is similar to the problem of the heteroclinic bifurcation estimation in type 1 PLL systems.

\section{Mathematical Model and Stability Analysis}
\begin{figure}[h]
	\centering
	\includegraphics[width=0.6\linewidth]{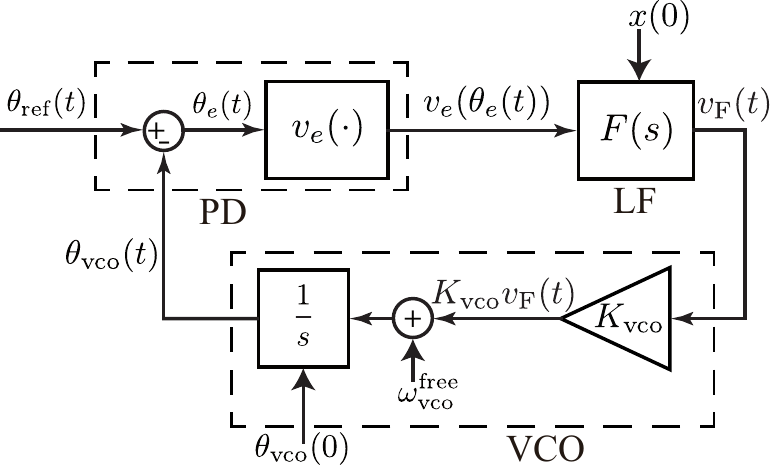}
	\caption{Baseband model of analog PLLs.}
	\label{fig:PLL-model}
\end{figure}

Consider analog PLL baseband model in Fig.~\ref{fig:PLL-model} \cite{Gardner-2005-book, Viterbi-1966, Best-2007, LeonovKYY-2012-TCASII, LeonovKYY-2015-SIGPRO}. 
Here $\theta_{\rm ref}(t) = \omega_{\rm ref}t + \theta_{\rm ref}(0)$ is a phase of the reference signal, 
a phase of the VCO is $\theta_{\rm vco}(t)$, 
$\theta_e(t) = \theta_{\rm ref}(t) - \theta_{\rm vco}(t)$ is a phase error.
A phase detector (PD) generates a signal $v_e(\theta_e(t))$
where $v_e(\cdot)$ is a characteristic of the phase detector. 
In the present paper, a piecewise-linear PD characteristic, which is continuous and corresponds to square waveforms of the reference and the VCO signals, is considered:
\begin{equation}\label{eq:piecewise-linear PD characteristic}
\begin{aligned}
&v_e(\theta_e)
=
\begin{cases}
& k\theta_e - 2\pi k m, 
\qquad -\frac{1}{k} + 2 \pi m \leq \theta_e(t) < \frac{1}{k} + 2 \pi m,\\
& -\frac{1}{\pi - \frac{1}{k}}\theta_e + \frac{1}{\pi - \frac{1}{k}}(\pi + 2\pi m),
\qquad\frac{1}{k} + 2 \pi m \leq \theta_e(t) < -\frac{1}{k} + 2 \pi (m + 1),
\end{cases}
\end{aligned}
\end{equation}
here $k>\frac{1}{\pi},\ m \in \mathbb{Z}$ (see Fig.~\ref{fig:triangular PD characteristic}).

\begin{figure}[h]
	\centering
	\includegraphics[width=0.4\linewidth]{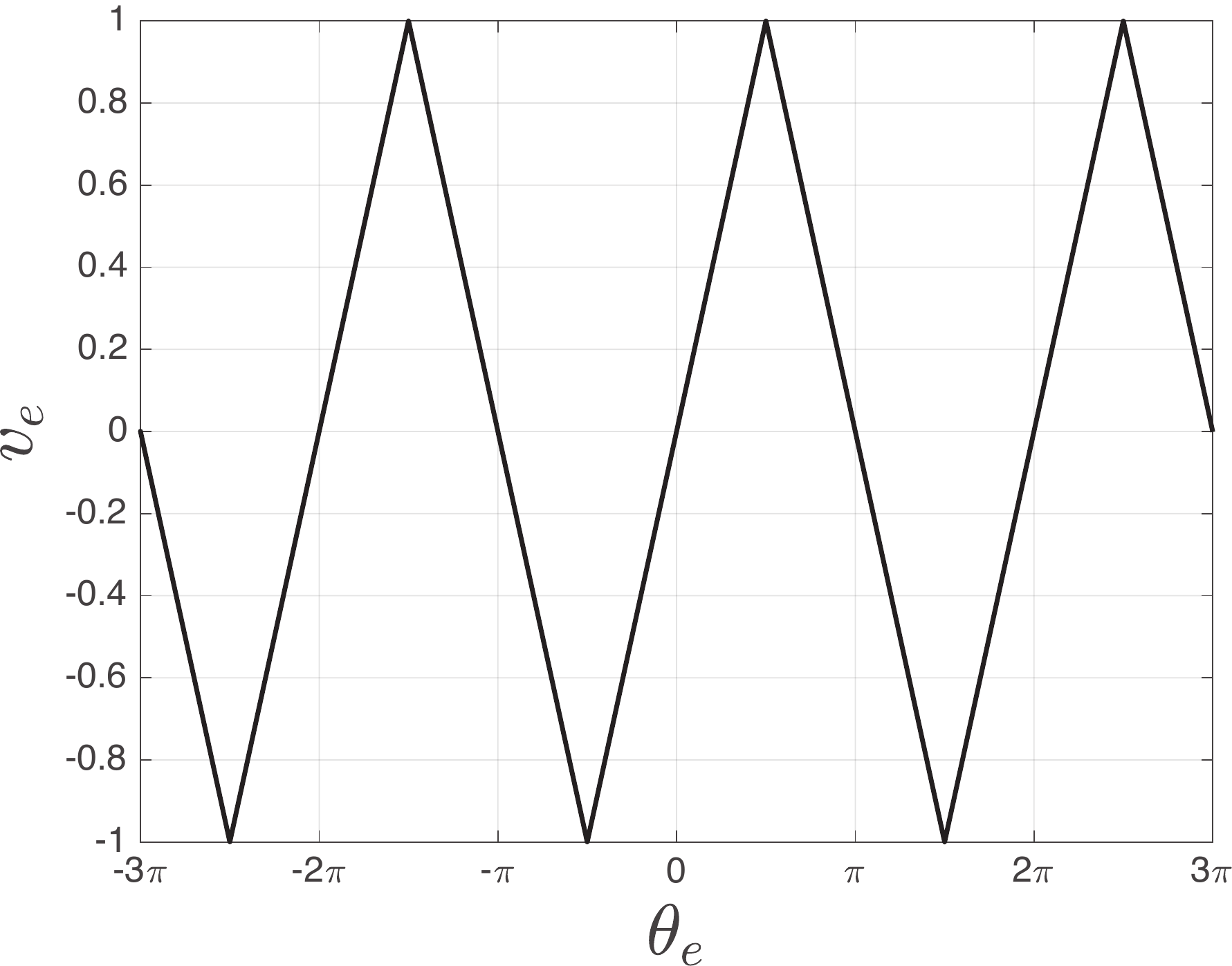}
	\caption{Triangular PD characteristic (piecewise-linear PD characteristic \eqref{eq:piecewise-linear PD characteristic} with $k = \frac{2}{\pi}$).}
	\label{fig:triangular PD characteristic}
\end{figure}

The state of the loop filter is represented by 
$x(t)\in\mathbb{R}$
and the transfer function is
\begin{equation*}
\begin{aligned}
& F(s) = \frac{1+s\tau_2}{s\tau_1},\quad
\tau_1>0,\ \tau_2>0.
\end{aligned}
\end{equation*}
The output of the loop filter $v_{\rm F}(t) = \frac{1}{\tau_1}(x(t) + \tau_2 v_e(\theta_e(t))$ is used to control the VCO frequency $\omega_{\rm vco}(t)$, which is proportional to the control voltage:
\begin{equation*}
	\begin{aligned}
		& \omega_{\rm vco}(t) = \dot\theta_{\rm vco}(t) = \omega_{\rm vco}^{\rm free} + K_{\rm vco} v_{\rm F}(t)
	\end{aligned}
\end{equation*}
where $K_{\rm vco}>0$ is a gain and $\omega_{\rm vco}^{\rm free}$ is a free-running frequency of the VCO. 

The behavior of PLL baseband model in the state space is described by a second-order nonlinear ODE:
\begin{equation}\label{eq:PLL-model}
\begin{aligned}
&\dot x = v_e(\theta_e), \\
&\dot \theta_e = \omega_e^{\rm free} - \frac{K_{\rm vco}}{\tau_1}\Big(x + \tau_2 v_e(\theta_e)\Big)
\end{aligned}
\end{equation}
where $\omega_e^{\rm free} = \omega_{\rm ref} - \omega_{\rm vco}^{\rm free}$ is a frequency error and $v_e(\theta_e)$ is defined in \eqref{eq:piecewise-linear PD characteristic}.
It is usually supposed that the reference frequency (hence, $\omega_e^{\rm free}$ too) can be abruptly changed and that the synchronization occurs between those changes.
Thus, existence of locked states, acquisition and transient processes after the reference frequency change are of interest.

\subsection{Local stability analysis}\label{sec:local analysis}

The PLL baseband model in Fig.~\ref{fig:PLL-model} is locked if the phase error $\theta_e(t)$ is constant. 
For the locked states of practically used PLLs, the loop filter state is constant
too and, thus, the locked states of model in Fig.~\ref{fig:PLL-model} correspond to the equilibria of model \eqref{eq:PLL-model} \cite{KuznetsovLYY-2015-IFAC-Ranges}.

\begin{definition}\cite{KuznetsovLYY-2015-IFAC-Ranges, LeonovKYY-2015-TCAS, BestKLYY-2016}
	A \emph{hold-in range} is the largest symmetric interval of frequency errors $|\omega_e^{\rm free}|$ such that an asymptotically stable equilibrium exists and varies continuously while $\omega_e^{\rm free}$ varies continuously within the interval.
	\label{def:hold-in}
\end{definition}

Observe that system \eqref{eq:PLL-model} is $2\pi$-periodic in $\theta_e$ and has an infinite number of equilibria $\left(\frac{\tau_1\omega_e^{\rm free}}{K_{\rm vco}},\ \pi m\right)$, $m\in\mathbb{Z}$.
The characteristic polynomial of system \eqref{eq:PLL-model} linearized at stationary states $\left(\frac{\tau_1\omega_e^{\rm free}}{K_{\rm vco}}, \pi m\right)$ is
\begin{equation*}
\begin{aligned}
&\chi(s) = s^2 + \frac{K_{\rm vco}\tau_2}{\tau_1}v_e^\prime(\pi m)s + \frac{K_{\rm vco}}{\tau_1}v_e^\prime(\pi m).
\end{aligned}
\end{equation*}
The nonlinearity $v_e(\theta_e)$ decreases $\left(v_e^\prime(\pi + 2\pi m) = -\frac{1}{\pi - \frac{1}{k}} < 0\right)$ for $\frac{1}{k} + 2 \pi m < \theta_e(t) < -\frac{1}{k} + 2 \pi (m + 1)$, and equilibria $\left(\frac{\tau_1\omega_e^{\rm free}}{K_{\rm vco}}, \pi + 2\pi m\right)$ are saddles.
The nonlinearity $v_e(\theta_e)$ increases ($v_e^\prime(2\pi m) = k > 0$) for $-\frac{1}{k} + 2 \pi m < \theta_e(t) < \frac{1}{k} + 2 \pi m$, and the equilibria $\left(\frac{\tau_1\omega_e^{\rm free}}{K_{\rm vco}}, 2\pi m\right)$ are asymptotically stable: 
\begin{itemlist}
	\item if $\frac{K_{\rm vco}\tau_2^2 k}{\tau_1} > 4$ then the equilibria $\left(\frac{\tau_1\omega_e^{\rm free}}{K_{\rm vco}}, 2\pi m\right)$ are asymptotically stable nodes, 
	\item if $\frac{K_{\rm vco}\tau_2^2 k}{\tau_1} = 4$ then the equilibria $\left(\frac{\tau_1\omega_e^{\rm free}}{K_{\rm vco}}, 2\pi m\right)$ are asymptotically stable degenerate nodes,
	\item if $\frac{K_{\rm vco}\tau_2^2 k}{\tau_1} < 4$ then the equilibria $\left(\frac{\tau_1\omega_e^{\rm free}}{K_{\rm vco}}, 2\pi m\right)$ are asymptotically stable focuses.
\end{itemlist}
Since an asymptotically stable equilibrium exists for any frequency error $\omega_e^{\rm free}$, the hold-in range of model \eqref{eq:PLL-model} is infinite for any loop parameters $K_{\rm vco}>0,\ \tau_1 > 0,\ \tau_2 > 0$.

\subsection{Global stability analysis}
\begin{definition}\cite{KuznetsovLYY-2015-IFAC-Ranges,LeonovKYY-2015-TCAS,BestKLYY-2016}
	A \emph{pull-in range} is the largest symmetric interval 
	of frequency errors $|\omega_e^{\rm free}|$ from the hold-in range such that an equilibrium is acquired for an arbitrary initial state.	
\end{definition}

In 1959, Andrew~J.~Viterbi applied the phase-plane analysis and stated that the second-order type 2 PLL models with sinusoidal PD characteristic have infinite (theoretically) hold-in and pull-in ranges for any loop parameters \cite[p.12]{Viterbi-1959}, \cite{Viterbi-1966}.
However, his proof was incomplete (see, e.g. discussion in \cite{AlexandrovKLNS-2015-IFAC}).
Later, Viterbi's statement was rigorously proved using the direct Lyapunov method ideas \cite{Bakaev-1963, AleksandrovKLNYY-2016-IFAC, KuznetsovLYY-2021-TCASII}.

To analyse the pull-in range of system \eqref{eq:PLL-model} with piecewise-linear PD characteristic, we apply the direct Lyapunov method and the corresponding theorem on global stability for the cylindrical phase space (see, e.g. \cite{LeonovK-2014-book, KuznetsovLYYKKRA-2020-ECC}). 
If there is a continuous function $V(x, \theta_e): \mathbb{R}^{n}\to\mathbb{R}$ such that

(i) $V(x, \theta_e + 2\pi)=V(x, \theta_e) \quad\forall x\in\mathbb{R}^{n-1}, \forall \theta_e\in\mathbb{R}$;

(ii) for any solution $(x(t), \theta_e(t))$ of system \eqref{eq:PLL-model} the function $V(x(t), \theta_e(t))$ is nonincreasing;

(iii) if $V(x(t), \theta_e(t))\equiv V(x(0), \theta_e(0))$, then $(x(t), \theta_e(t))\equiv~(x(0), \theta_e(0))$;

(iv) $V(x, \theta_e)+\theta_e^2\to+\infty$ as $||x||+|\theta_e|\to+\infty$ \\ 
then  
any trajectory of system \eqref{eq:PLL-model} tends to an equilibrium
(for brevity, we shall call such systems \textit{globally stable}).

Consider the following Lyapunov function:
\begin{equation}\label{eq:Lyapunov function}
\begin{aligned}
&V(x, \theta_e)= \frac{K_{\rm vco}}{2\tau_1}\left(x-\frac{\tau_1\omega_e^{\rm free}}{K_{\rm vco}}\right)^2 +
\int\limits_0^{\theta_e} v_e(\sigma) d\sigma.
\end{aligned}
\end{equation}
Its derivative along the trajectories of system \eqref{eq:PLL-model} is
\begin{equation*}
\begin{aligned}
&\dot{V}(x, \theta_e)=-
\frac{K_{\rm vco}\tau_2}{\tau_1}v_e^2(\theta_e) < 0 \quad\forall \theta_e \ne\pi m,\ m\in\mathbb{Z}.
\end{aligned}
\end{equation*}
Since the derivative along any solution other than stationary states is not identically zero, system \eqref{eq:PLL-model} is globally stable for any $\omega_e^{\rm free}$ and, hence, the pull-in range is infinite.

In 1981, William~F.~Egan conjectured \cite[p.176]{Egan-1981-book} that a higher-order {\it type 2 PLL with an infinite hold-in range also has an infinite pull-in range}, and supported it with some third-order PLL implementations (see also \cite[p.161]{Egan-2007-book}).
However, this conjecture is not valid in general and corresponding counterexamples were recently provided in \cite{KuznetsovLYY-2021-TCASII}.

Notice that a similar conjecture on the pull-in range for the second-order type 1 PLLs is known as \textit{the Kapranov conjecture} \cite{Kapranov-1956}, where it is supposed that the global stability of the corresponding model is determined by the birth of self-excited oscillations only, not hidden ones \cite{LeonovK-2013-IJBC, ChenKLM-2017-IJBC}.
Discussions of counterexamples to the Kapranov conjecture can be found in \cite{KuznetsovLYY-2017-CNSNS, Kuznetsov-2020-TiSU}.

\section{The lock-in range of second-order type 2 analog PLL with piecewise-linear PD characteristic}\label{sec:triangular_lock-in}

Although a PLL model can be globally stable with infinite pull-in range, the acquisition process can take long time.
To decrease the synchronization time, a lock-in range concept is frequently exploited \cite{Gardner-2005-book, Kolumban-2005, Best-2007}.

\begin{definition}\cite{KuznetsovLYY-2015-IFAC-Ranges, LeonovKYY-2015-TCAS, BestKLYY-2016}
	A \emph{lock-in range} is the largest interval of frequency errors $|\omega_e^{\rm free}|$ from the pull-in range such that the PLL model being in an equilibrium, after any abrupt change of $\omega_e^{\rm free}$ within the interval acquires an equilibrium without cycle slipping ($\sup\limits_{t>0} |\theta_e(0) - \theta_e(t)| < 2\pi$).	
\end{definition}

\begin{remark}\label{remark:cycle slipping}
Sometimes the upper limit is considered in the cycle slipping definition instead of the supremum: $\limsup\limits_{t\to+\infty} |\theta_e(0)-\theta_e(t)|\ge2\pi$.
For any $\omega_e^{\rm free}$ the following inequality is valid: $\sup\limits_{t>0}|\theta_e(0)~-~\theta_e(t)|~\ge~ \limsup\limits_{t\to+\infty} |\theta_e(0) - \theta_e(t)|$. 
However, bifurcation values determining the lock-in range $[0, \omega_l)$ are the same for both definitions of cycle slipping (see Fig.~\ref{fig:suplimsup}).
\end{remark}
\begin{figure*}[h]
	\begin{minipage}[h]{\linewidth}
		\centering\includegraphics[width=0.45\linewidth]{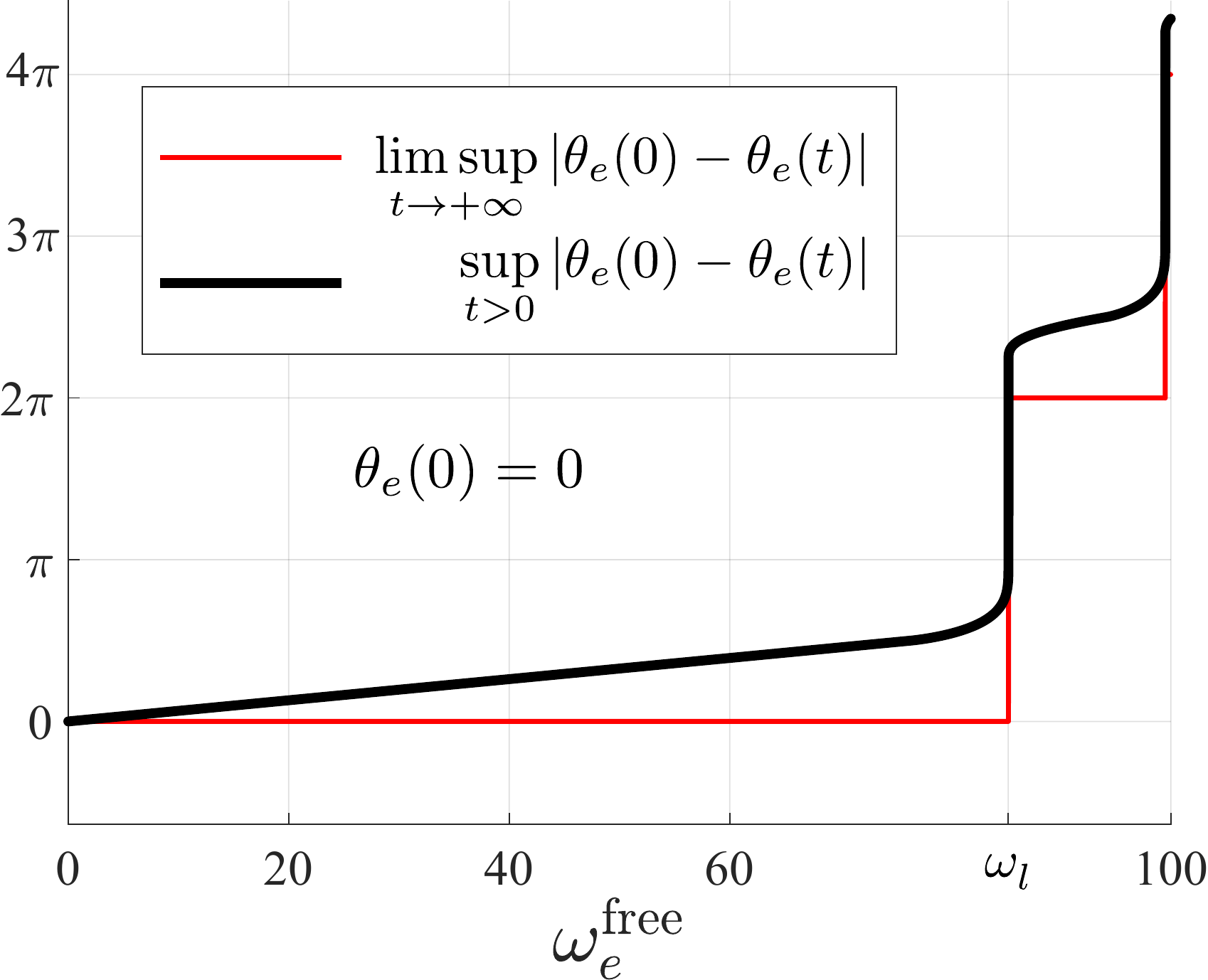}
	\end{minipage}
	\caption{Comparison of cycle slipping definitions (see Remark~\ref{remark:cycle slipping}) for model \eqref{eq:PLL-model} with parameters
		$\tau_1 = 0.0633$,
		$\tau_2 = 0.0225$,
		$K_{\rm vco} = 250$.
	}
	\label{fig:suplimsup}
\end{figure*}

From a mathematical point of view, system \eqref{eq:PLL-model} can initially be in an unstable equilibrium (at one of the saddles) or can acquire it by a separatrix after a change of $\omega_e^{\rm free}$ (see~\cite{KuznetsovBAYY-2019, KuznetsovLYYK-2020-IFACWC}). 
Corresponding behavior is not observed in practice: system state is disturbed by noise and can't remain in unstable equilibrium.
In this paper, two cycle-slipping-related characteristics of the system are considered: 
\textit{the lock-in range} $|\omega_e^{\rm free}|\in [0, \omega_l)$ where the equilibria are considered to be stable and 
\textit{the conservative lock-in range} $|\omega_e^{\rm free}|\in[0, \omega_l^c) \subset [0, \omega_l)$ which takes into account the unstable behavior described above.

For the considered model boundary values $\omega_l$ and $\omega_l^c$ are determined by \textit{cycle slipping bifurcation}. It happens when the system being in an equilibrium state is exposed to an abrupt change of $\omega_e^{\rm free}$, and the corresponding trajectory of the system after the switch tends to the nearest unstable equilibrium by the corresponding saddle separatrix.
In other words, $\sup\limits_{t>0} |\theta_e(0) - \theta_e(t)| = \limsup\limits_{t \to +\infty} |\theta_e(0) - \theta_e(t)| = \pi$ for $\theta_e(0) = 2\pi$ (see Fig.~\ref{fig:lock-in illustration}, lower left picture) and $\sup\limits_{t>0} |\theta_e(0) - \theta_e(t)| = \limsup\limits_{t\to+\infty} |\theta_e(0) - \theta_e(t)| = 2\pi$ for $\theta_e(0) = 3\pi$ (see Fig.~\ref{fig:lock-in illustration}, upper right picture).
For a larger $\omega_e^{\rm free}$ supremum 
$\sup\limits_{t>0} |\theta_e(0) - \theta_e(t)| > 2\pi$ and cycle slipping occurs.
Since the lock-in range is defined as a half-open interval, boundary values $\omega_e^{\rm free} = \omega_l$ and $\omega_e^{\rm free} = \omega_l^c$ are not included in it.

\begin{figure*}[h]
	\centering
	\begin{minipage}[h]{0.4\linewidth}
		\includegraphics[width=\linewidth]{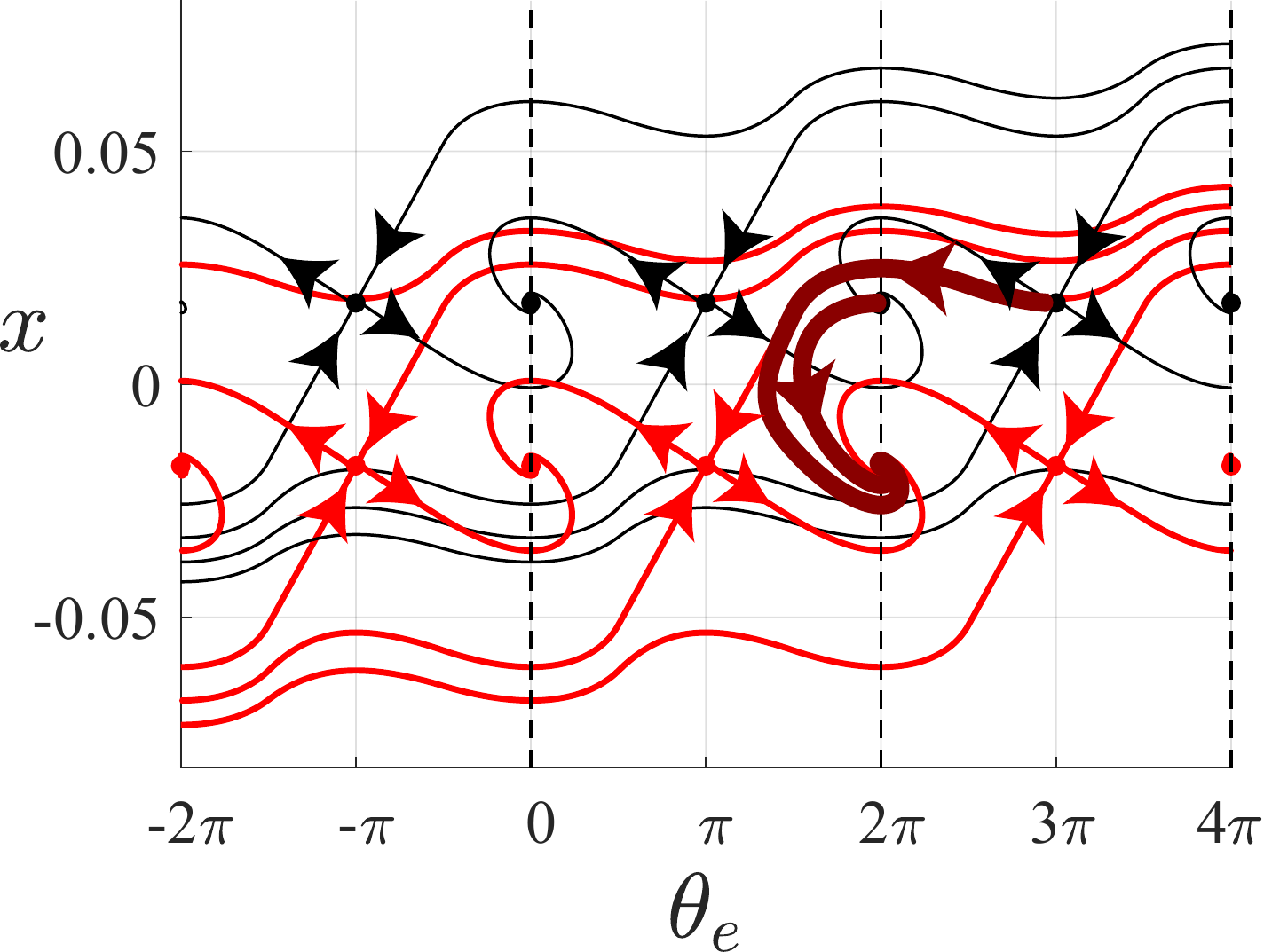}         	
	\end{minipage}
	$\qquad$
	\begin{minipage}[h]{0.4\linewidth}
		\includegraphics[width=\linewidth]{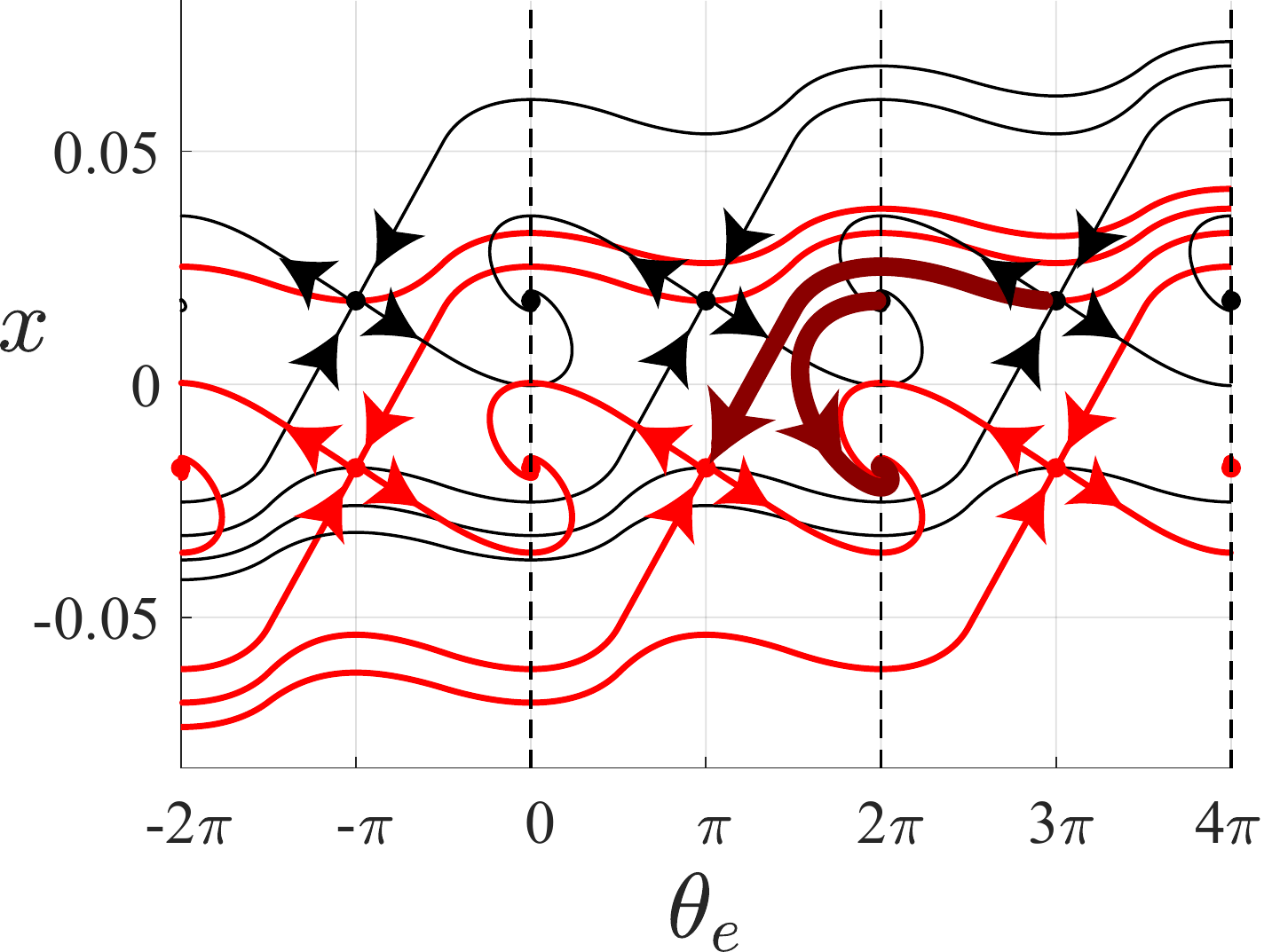}
	\end{minipage}
	
	\begin{minipage}[h]{0.4\linewidth}
		\includegraphics[width=\linewidth]{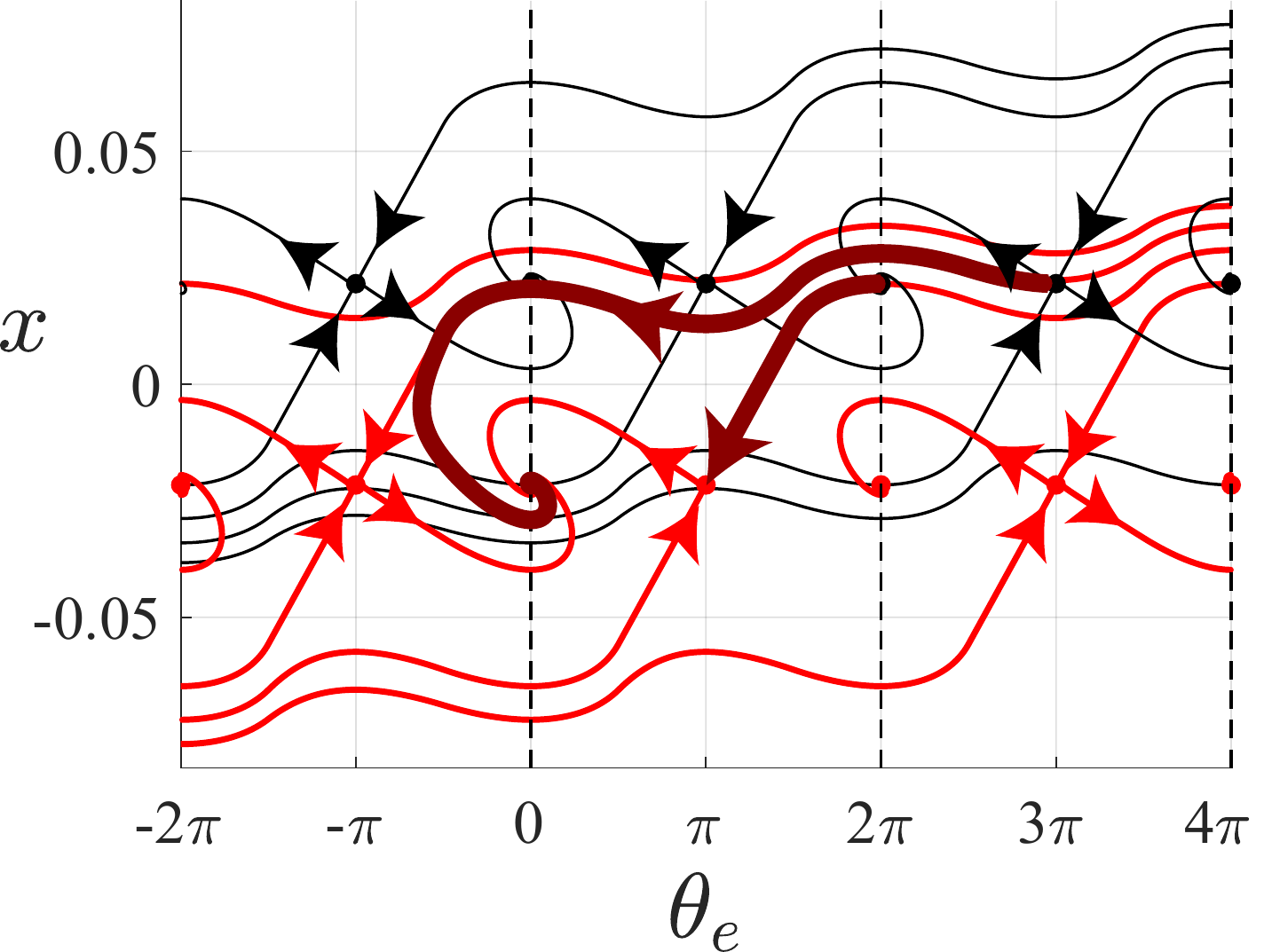}       	
	\end{minipage}
	$\qquad$
	\begin{minipage}[h]{0.4\linewidth}
		\includegraphics[width=\linewidth]{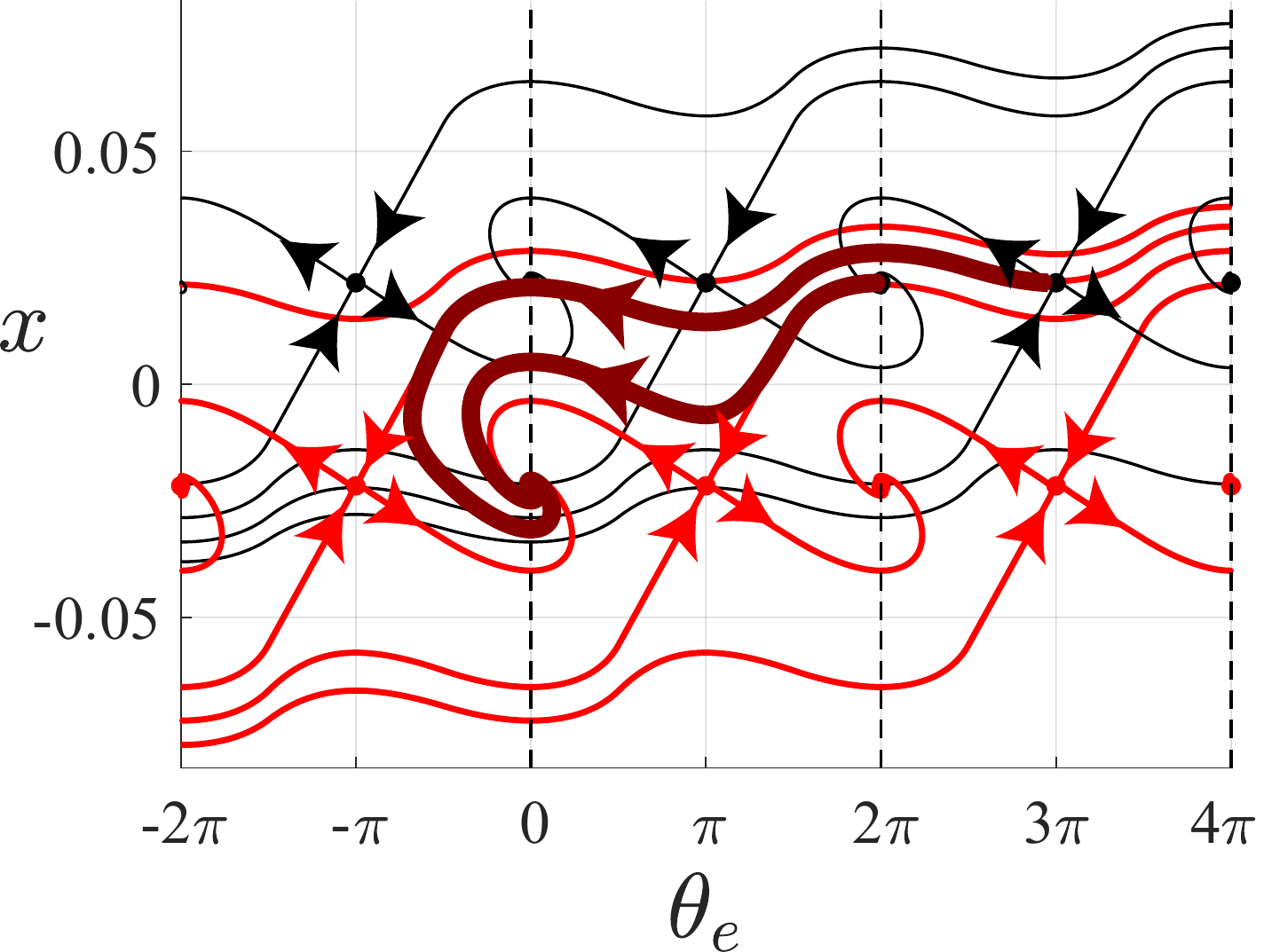}
	\end{minipage}
	\caption{Phase portraits for model \eqref{eq:PLL-model} with the following parameters:
		$F(s)= \frac{1+s\tau_2}{s\tau_1}$,
		$\tau_1 = 0.0633$,
		$\tau_2 = 0.0225$,
		$K_{\rm vco}=250$.
		Black dots are equilibria of the model with positive $\omega_e^{\rm free}=|\omega|$.
		Red color is for the model with negative $\omega_e^{\rm free}=-|\omega|$.
		Separatrices pass in and out of the saddles equilibria.
		Upper left subfigure: $\omega = 69 < \omega_l^c$,
		upper right subfigure: $\omega = \omega_l^c \approx 70.79$ (evaluated by Theorem~2),
		lower left subfigure: $\omega = \omega_l \approx 85.27$ (evaluated by Theorem~1),
		lower right subfigure: $\omega = 86 > \omega_l$.
	}
	\label{fig:lock-in illustration}
\end{figure*}

In practice, the lock-in range can be estimated in the following way.
Without loss of generality we can fix $\omega_{\rm vco}^{\rm free}$ and vary $\omega_{\rm ref}$ only.
Let initially $\omega_e^{\rm free} = \omega_{\rm ref} - \omega_{\rm vco}^{\rm free} = 0$ and the system is in a stable equilibrium.
Then we abruptly increase the reference frequency by sufficiently small frequency step $\Delta\omega>0$ (i.e., the reference frequency becomes $\omega_{\rm ref} =  \omega_{\rm vco}^{\rm free} + \Delta\omega$) and observe whether corresponding transient process converges to a locked state without cycle slipping (see Fig.~\ref{fig:lock-in-procedure}). 
After that we abruptly decrease the reference frequency by $2\Delta\omega$ (i.e., the reference frequency becomes $\omega_{\rm ref} =  \omega_{\rm vco}^{\rm free} - \Delta\omega$). 
If the transient process converges to the locked state without cycle slipping, then
$[0, \Delta\omega) \subset [0, \omega_l)$.
Frequency step $\Delta\omega>0$ should be increased until cycle slipping occurs.

\begin{figure}[!h]
	\centering
	\includegraphics[width=0.6\textwidth]{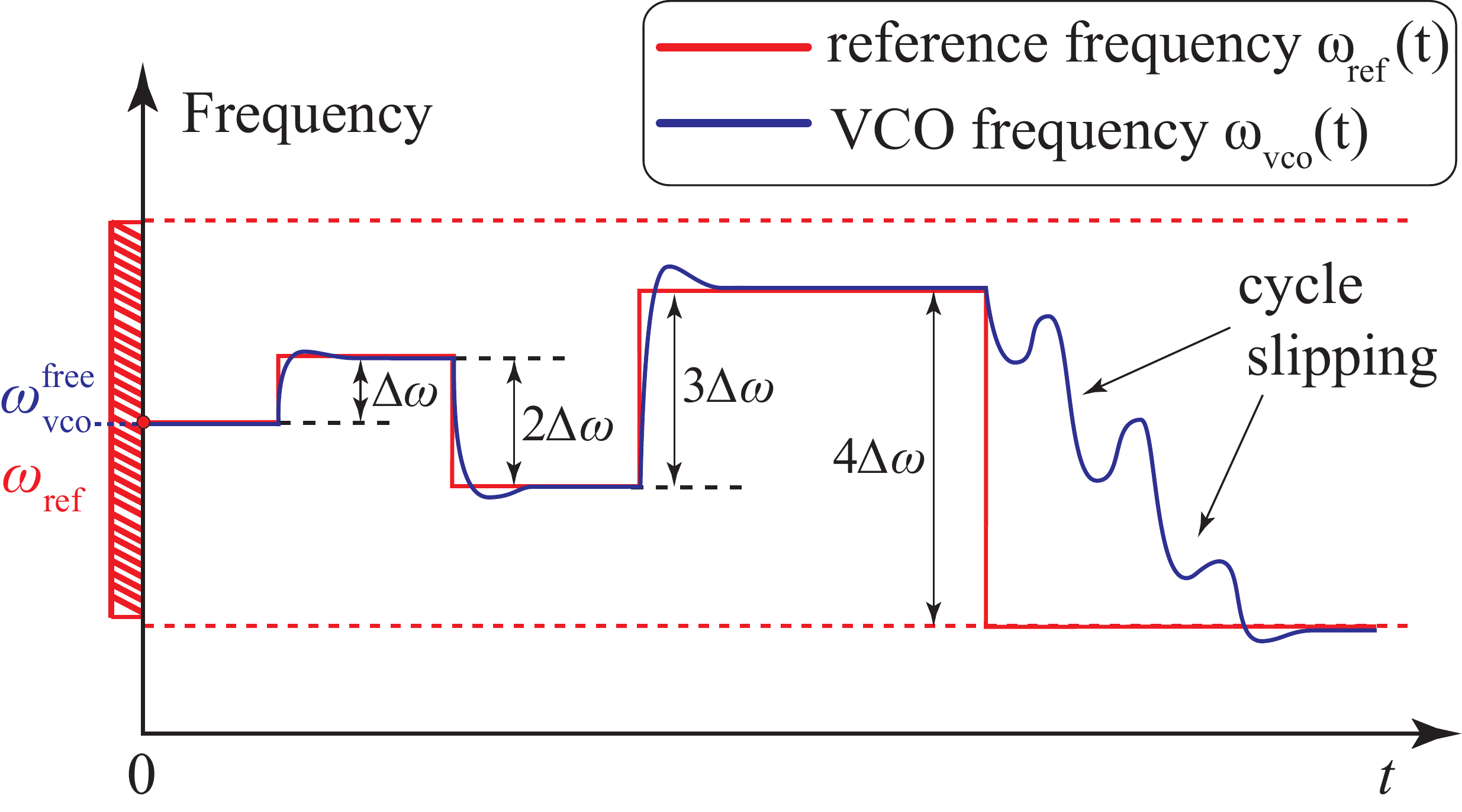}
	\caption{The lock-in range calculation.}
	\label{fig:lock-in procedure}
\end{figure}

Using changes of variables we represent system \eqref{eq:PLL-model} as the first-order differential equation \cite{Belyustina-1959, HuqueS-2011-exact} and following \cite{AleksandrovKLNYY-2016-IFAC, KuznetsovBAYY-2019} we formulate and prove theorems providing exact values for the lock-in range and for the conservative lock-in range.

\begin{theorem}\label{theorem: lock-in stable}
	The lock-in frequency of model \eqref{eq:PLL-model} with the piecewise-linear PD characteristic \eqref{eq:piecewise-linear PD characteristic} is
	\begin{equation}\label{eq:lock-in_exact_formula}
	\begin{aligned}
	&\omega_l
	=
	\begin{cases}
	& \frac{a\sqrt{\pi}}{2\tau_2}\Big(\frac{c + b}{c - b}\Big)^{\frac{a}{2b}},
	\quad a^2k>4,\\
	&\frac{a\sqrt{\pi}}{2\tau_2} 
	\exp(\frac{a}{2\sqrt{\pi}}),
	\quad a^2k = 4,\\
	& \frac{a\sqrt{\pi}}{2\tau_2}
	\exp\Big({\frac{a}{b}
		\arctan\frac{b}{c}}\Big),
	\quad a^2k < 4
	\end{cases}
	\end{aligned}
	\end{equation}
	where
	\begin{equation}\label{eq:abc}
	\begin{aligned}
	& a = \sqrt{\frac{K_{\rm {vco}}}{\tau_1}}\tau_2, \quad
	b = \sqrt{|a^2 - \frac{4}{k}|}, \quad
	c = \sqrt{a^2 + 4(\pi - \frac{1}{k})}.
	\end{aligned}
	\end{equation}
\end{theorem}

\begin{theorem}\label{theorem: lock-in unstable}
	The conservative lock-in frequency of model \eqref{eq:PLL-model} with piecewise-linear PD characteristic \eqref{eq:piecewise-linear PD characteristic} is
	\begin{equation}\label{eq:conservative_lock-in_exact_formula}
	\begin{aligned}
	&\omega_l^c
	=
	\ \frac{1}{2}
	\sqrt{\frac{K_{\rm vco}
			(d + \frac{c - a}{2})^{\frac{c - a}{c}}
			(d - \frac{c + a}{2})^{\frac{c + a}{c}}
		}{\tau_1}},
	\end{aligned}
	\end{equation}
	where $a,\ b$ and $c$ are evaluated by \eqref{eq:abc}, and $d$ is the unique solution of one of the equations:
	\begin{equation}\label{eq:equations-for-y23}
	\begin{aligned}
	& \begin{cases}
	&
	(d - \frac{a - b}{2})^{\frac{b - a}{b}}
	(d - \frac{a + b}{2})^{\frac{b + a}{b}}
	=
	\pi
	(\frac{c + b}{c - b})^{\frac{a}{b}},\
	d > \frac{a + b}{2},
	\quad a^2k>4,\\
	&d 
	= 
	\frac{a}{2}
	\Big(1 + \dfrac{1}{W(\frac{a}{2\sqrt{\pi}}\exp({-\frac{a}{2\sqrt{\pi}}}))}\Big),
	\quad a^2k=4,\\
	&
	\Big(
	d^2 - ad + \frac{1}{k}
	\Big)
	\exp
	\Big(
	\frac{2a}{b} \arctan\frac{b}{a - 2d}
	\Big)
	=
	\pi
	\exp
	\Big(\frac{2a}{b}
	\arctan\frac{b}{c}
	\Big),\
	d > \frac{a}{2},
	\quad a^2k<4.
	\end{cases}
	\end{aligned}
	\end{equation}
	Here $W(x)$ is the Lambert $W$ function.
\end{theorem}

\begin{proof}[Proof of Theorem~\ref{theorem: lock-in stable} and Theorem~\ref{theorem: lock-in unstable}]
	The proof given in Appendix~A is based on the fact that system \eqref{eq:PLL-model} is piecewise-linear and can be integrated analytically on the linear segments.
\end{proof}

	Notice that $\omega_l$ and $\omega_l^c$ are continuous functions
of variable $k$ (as $a$ is fixed): the cases $a^2k>4$ and $a^2k<4$ in formulae \eqref{eq:lock-in_exact_formula}, \eqref{eq:conservative_lock-in_exact_formula} approach the case $a^2k = 4$ as $k \to \frac{4}{a^2}$ (as $b\to 0$).

\section{Conclusions} \noindent 
In this work, the exact formulae for the lock-in range and the conservative lock-in range for the second-order type 2 PLL with a piecewise-linear phase detector characteristic were derived. 
In engineering literature, the following approximate estimate for the lock-in range can be found:
\begin{equation}\label{eq:Gardner_lock-in_estimate} 
	\begin{aligned}
		& \omega_l \approx \frac{ K_{\rm vco}\tau_2}{\tau_1}
	\end{aligned}
\end{equation} 
(see \cite[p.69]{Best-2007} where $\omega_l \approx \pi \zeta \omega_n$, $\omega_n = \sqrt{K_dK_{\rm vco}/\tau_1}$, $\zeta = \omega_n\tau_2/2$, $K_d = \frac{2}{\pi}$, 
and 
\cite[p.187]{Gardner-2005-book} where $K_d = 1, K_o = K_{\rm vco}$).
However, estimate \eqref{eq:Gardner_lock-in_estimate} intersects the exact lock-in frequency value \eqref{eq:lock-in_exact_formula} for some values of parameters.
Taking into account that for type 2 PLLs a pull-out frequency\footnote{
	In 1966, such concept as pull-out frequency was introduced by F.~Gardner \cite[p.37]{Gardner-1966}.
	In the literature, the following explanations of the pull-out frequency $\omega_{\rm po}$ can be found: 
	``some frequency-step limit below which the loop does not slip cycles but remains in lock'' \cite[p.37]{Gardner-1966}, \cite[p.116]{Gardner-2005-book},
	``the maximum value of the input reference frequency step that can be applied to a phase-locked PLL, yet the loop is able to relock without slipping a cycle'' \cite{Stensby-1997, HuqueS-2011-exact, HuqueS-2013} (see also \cite[p.59]{Best-2007}).
	Since using a linear change of variables the value $\omega_e^{\rm free}$ can be excluded from the type 2 PLL systems \cite{KuznetsovLYY-2021-TCASII}, such concept is consistent for them and corresponds to the lock-in frequency in the following way: $\omega_{\rm po} = 2\omega_l$.
	However, equilibria of type 1 PLLs depend on the frequency error $\omega_e^{\rm free}$ and, hence, the correct pull-out frequency definition should take into account the initial value of the frequency error corresponding to the locked state.
}
$\omega_{\rm po}$ is twice the value of the lock-in frequency, one more approximate estimate for the lock-in range is exploited:
\begin{equation*} 
	\begin{aligned}
		& \omega_l 
		\approx 
		0.7995\sqrt{\frac{2K_{\rm vco}}{\pi\tau_1}}
		+
		1.23\frac{\tau_2K_{\rm vco}}{\pi\tau_1}
	\end{aligned}
\end{equation*}
(see \cite[p.84]{Best-2007} where
$2\omega_l = \omega_{\rm po} 
\approx 
2.46\omega_n(\zeta + 0.65)$,
$\omega_n = \sqrt{K_dK_{\rm vco}/\tau_1},\ \zeta = \omega_n\tau_2/2,\ K_d = \frac{2}{\pi}$).

A.S.~Huque and J.~Stensby analysed system \eqref{eq:PLL-model} with a triangular PD characteristic [the piecewise-linear PD characteristic \eqref{eq:piecewise-linear PD characteristic} with $k = \frac{2}{\pi}$] in \cite{HuqueS-2011-exact, Huque-2011-PhD}.
However, in those works the global stability of system \eqref{eq:PLL-model} was not analysed.
In these works, the following formula for a pull-out frequency was derived:
\begin{equation}\label{eq:Stensby_formula triangular} 
\begin{aligned}
& \omega_{\rm po} 
= 
\frac{a^2}{\tau_2} \exp\Big(
\frac{1}{2} \ln |m_{-}^2 - m_{-} + a^\prime| - \frac{1}{\sqrt{4 a^\prime - 1}} \arctan \left(\frac{1 - 2m_{-}}{\sqrt{4 a^\prime - 1}}\right) + \frac{\pi}{2\sqrt{4 a^\prime - 1}} \Big)
\end{aligned}
\end{equation}
where 
$a^\prime = \frac{\pi}{2a^2},\ 
m_- = \frac{1}{2}(1 - \sqrt{4 a^\prime + 1})$.
For $a^2 < 2\pi$ the lock-in frequency $\omega_l = \frac{1}{2}\omega_{\rm po}$ with $\omega_{\rm po}$ from \eqref{eq:Stensby_formula triangular} coincides with the corresponding case in \eqref{eq:lock-in_exact_formula}, however for $a^2 \ge 2\pi$ formula \eqref{eq:Stensby_formula triangular} is formally not applicable and equations \eqref{eq:lock-in_exact_formula} should be used.

It's important to note that obtained lock-in range formula \eqref{eq:lock-in_exact_formula} is also a lower analytical estimate for the lock-in range of the second-order type 2 PLL with a sinusoidal PD characteristic.
For these systems several engineering estimates are known (see, e.g., \cite[p.117]{Gardner-2005-book} and \cite{HuqueS-2013} for the pull-out range estimates, and \cite[p.187]{Gardner-2005-book}, \cite[p.3748]{Kolumban-2005}, \cite[p.67]{Best-2007}, \cite{BestKLYY-2016}, \cite[p.18]{Best-2018} for the lock-in range estimates).

The further development of such systems analysis is connected with consideration of higher-order loop filters and discontinuous phase detector characteristics for revealing hidden oscillations and providing the global stability \cite{ZhuWEK-2020, KuznetsovMYY-2021-TCAS}.

\nonumsection{Acknowledgments} \noindent 
The work is funded by Team Finland Knowledge programme (163/83/2021) and by the Ministry of Science and Higher Education of the Russian Federation as part of World-class Research Center program: Advanced Digital Technologies (contract No. 075-15-2020-934 dated 17.11.2020). 
Z.~Wei acknowledges support from the National Natural Science Foundation of China (Nos. 12172340 and 11772306).

\appendix{Proof of Theorem~\ref{theorem: lock-in stable} and Theorem~\ref{theorem: lock-in unstable}}\label{sec:appendix proof}
\begin{proof}[Proof of Theorem~\ref{theorem: lock-in stable} and Theorem~\ref{theorem: lock-in unstable}]
	Let's find the lock-in range of model \eqref{eq:PLL-model} with piecewise-linear PD characteristic \eqref{eq:piecewise-linear PD characteristic}.
	As it was noted in section~\ref{sec:triangular_lock-in}, the lock-in frequency can be determined by such an abrupt change of $\omega_e^{\rm free}$ that the corresponding trajectory tends to the nearest unstable equilibrium (by the corresponding separatrix).
	Suppose that initially the frequency error was equal to $\omega_e^{\rm free} = -\omega < 0$, but then changed to $\omega_e^{\rm free} = \omega > 0$.
	Hence, initially the system is in equilibrium
	$x^{\rm eq} = -\frac{\tau_1\omega}{K_{\rm vco}},\quad
	\theta_e^{\rm eq} = 0$, but after the switch the corresponding trajectory tends to
	$x^{\rm eq} = \frac{\tau_1\omega}{K_{\rm vco}},\quad
	\theta_e^{\rm eq} = 0$ without cycle slipping if $\omega < \omega_l$.
	
	Such $\omega_l$ is determined by such frequency error $\omega_e^{\rm free}$ that a trajectory being in stable equilibrium (before the switch)
	$x^{\rm eq} = -\frac{\tau_1\omega_l}{K_{\rm vco}},\quad
	\theta_e^{\rm eq} = 0$
	tends to saddle equilibrium (after the switch)
	$x^{\rm eq} = \frac{\tau_1\omega_l}{K_{\rm vco}},\quad
	\theta_e^{\rm eq} = \pi$
	by the corresponding separatrix.	
	Thus, the lock-in frequency $\omega_l$ corresponds to the case
	\begin{equation}\label{eq:lock-in relations with Q}
	\begin{aligned}
	&	-\frac{\tau_1\omega_l}{K_{\rm vco}} = Q(0, \omega_l)
	\end{aligned}
	\end{equation}
	where $\frac{\tau_1\omega_e^{\rm free}}{K_{\rm vco}}$ is $x$-coordinate of equilibrium of model \eqref{eq:PLL-model} and $x = Q(\theta_e, \omega_e^{\rm free})$ is the lower separatrix of saddle equilibrium $(\frac{\tau_1\omega_e^{\rm free}}{K_{\rm vco}}, \pi)$ (see Fig.~\ref{fig:lock-in illustration}).
	
	After the change of variables
	$\tau = \sqrt{\frac{K_{\rm vco}}{\tau_1}}t$,  
	$y = \sqrt{\frac{\tau_1}{K_{\rm vco}}}\omega_e^{\rm free}  - \sqrt{\frac{K_{\rm vco}}{\tau_1}} (x + \tau_2v_e(\theta_e))$,
	for $\theta_e(t)\in (-\frac{1}{k} + 2 \pi n,\ \frac{1}{k} + 2 \pi n)$ and $ \theta_e(t)\in(\frac{1}{k} + 2 \pi n,\ -\frac{1}{k} + 2 \pi (n + 1))$
	system \eqref{eq:PLL-model} is represented as follows:
	\begin{equation}\label{eq:PLL-triangular-after_change_of_variables}
	\begin{aligned}
	&\dot y = - av_e^\prime(\theta_e)y - v_e(\theta_e),\\
	&\dot{\theta}_e = y,
	\end{aligned}
	\end{equation}
	where $a = \tau_2\sqrt{\frac{K_{\rm {vco}}}{\tau_1}}$.
	
	Upper separatrix $y = S(\theta_e)$ of the phase plane of \eqref{eq:PLL-triangular-after_change_of_variables} corresponds to separatrix $x = Q(\theta_e, \omega_e^{\rm free})$ from \eqref{eq:PLL-model} (see Fig.~\ref{fig:SeparatricesEquivalent}) and has the form
	\begin{equation*}
	\begin{aligned}
	&S(\theta_e) = \sqrt{\frac{\tau_1}{K_{\rm vco}}}\omega_e^{\rm free} - \sqrt{\frac{K_{\rm vco}}{\tau_1}}(Q(\theta_e, \omega_e^{\rm free}) + \tau_2v_e(\theta_e)).
	\end{aligned}
	\end{equation*}
	\begin{figure}[!h]
		\centering
		\includegraphics[width=\textwidth]{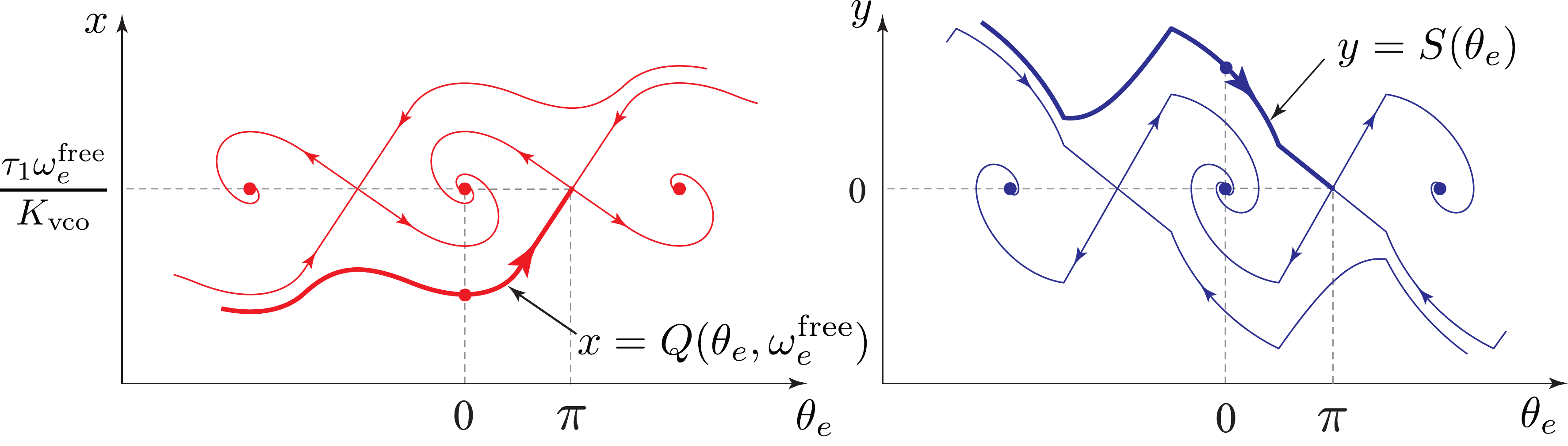}
		\caption{Phase plane portraits of \eqref{eq:PLL-model} and \eqref{eq:PLL-triangular-after_change_of_variables}.}
		\label{fig:SeparatricesEquivalent}
	\end{figure}
	Thus, relation \eqref{eq:lock-in relations with Q} takes the form
	\begin{equation*}
	\begin{aligned}
	&-\frac{\tau_1\omega_l}{K_{\rm vco}} = \frac{\tau_1\omega_l}{K_{\rm vco}} - \sqrt{\frac{\tau_1}{K_{\rm vco}}}S(0).
	\end{aligned}
	\end{equation*}
	Hence, $\omega_l = \frac{a}{2\tau_2}S(0)$.
	Analogously to the phase plane analysis for $\omega_l$, we get the following formula for the conservative lock-in frequency\footnote{
		To be more precise, for the conservative lock-in frequency it should be formally written $\omega_l^c = \min(\omega_l,\ \frac{a}{2\tau_2}S(-\pi))$, however, 
		$S(0) - S(-\pi) = -\sqrt{\frac{K_{\rm vco}}{\tau_1}} (Q(0, \omega_e^{\rm free}) - Q(-\pi, \omega_e^{\rm free})) > 0$ because $\dot x = v_e(\theta_e) < 0 $ as $\theta_e\in[-\pi,0]$.}:
	$\omega_l^c = \frac{a}{2\tau_2}S(-\pi)$.
	Denote
	\begin{equation*}
	\begin{aligned}
	&y_l = S(0),
	\qquad
	y_l^c = S(-\pi)
	\end{aligned}
	\end{equation*}
	and get the formulae for $\omega_l$ and $\omega_l^c$:
	\begin{equation}\label{eq:omega_l_formula-appendix}
	\begin{aligned}
	&\omega_l = \frac{a}{2\tau_2}y_l,
	\end{aligned}
	\end{equation}
	\begin{equation}\label{eq:omega_l^c_formula-appendix}
	\begin{aligned}
	&\omega_l^c = \frac{a}{2\tau_2}y_l^c.
	\end{aligned}
	\end{equation}
	
	The computation of $y_l$ and $y_l^c$ from formulae  \eqref{eq:omega_l_formula-appendix}, \eqref{eq:omega_l^c_formula-appendix} consists of the following stages. 
	Let's divide the phase plane to the following domains:
	\begin{itemize}
		\item I: $\{(y,\ \theta_e) \mid \frac{1}{k}\le \theta_e \le \pi$; $\theta_e,y \in \mathbb{R}\}$,
		\item II: $\{(y,\ \theta_e) \mid -\frac{1}{k}\le \theta_e \le \frac{1}{k}$; $\theta_e,y \in \mathbb{R}\}$,
		\item III: $\{(y,\ \theta_e) \mid -\pi \le \theta_e \le -\frac{1}{k}$; $\theta_e,y \in \mathbb{R}\}$.
	\end{itemize}
	In the open domains, system \eqref{eq:PLL-triangular-after_change_of_variables} is a linear one and can be integrated analytically.
	Firstly, we compute $S(\frac{1}{k})$, which is possible due to the continuity of \eqref{eq:PLL-model}.
	Using the obtained value as the initial data of the Cauchy problem and finding its solution in the domain II, we can compute $y_l = S(0)$ and $S(-\frac{1}{k})$. 
	Here exist three cases depending on the stable equilibrium type: an asymptotically stable focus, an asymptotically stable node, and an asymptotically stable degenerated node.
	For every case described above we perform separate computations.
	Using the obtained value as the initial data of the Cauchy problem and finding its solution in the domain III, we can compute $y_l^c = S(-\pi)$ (see Fig.~\ref{fig:separatrix integration}).
	
	\begin{figure}[h]
		\centering
		\includegraphics[width=0.7\linewidth]{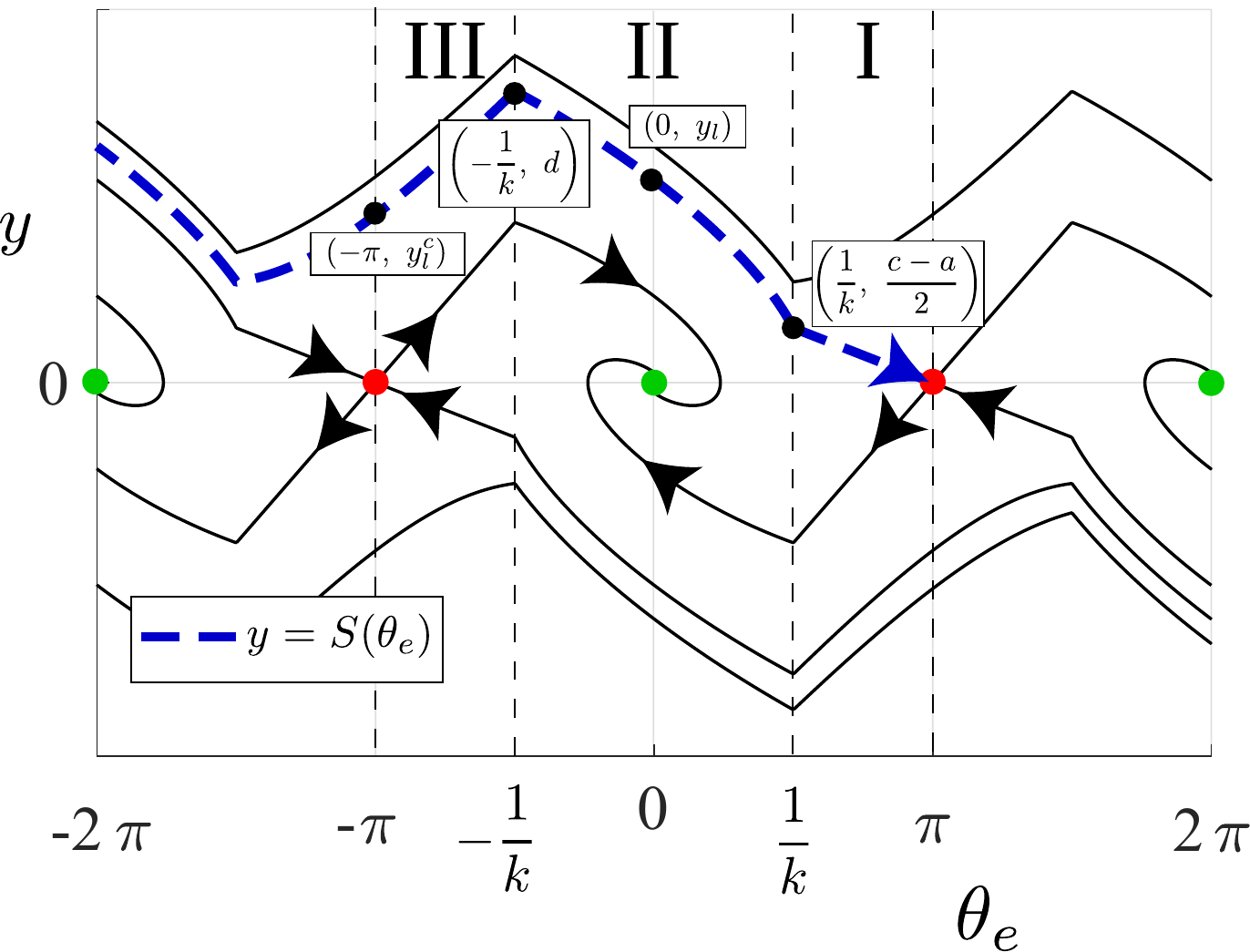}
		\caption{The separatrix integration.
		Firstly, we compute $S(\frac{1}{k})$ and use it as the initial data of the Cauchy problem.
		Secondly, finding its solution in the domain II, we compute $y_l = S(0)$, which is used for the lock-in frequency $\omega_l$ computation (see \eqref{eq:omega_l_formula-appendix}), and $S(-\frac{1}{k})$. 
		Finally, we use $S(-\frac{1}{k})$ as the initial data of the Cauchy problem and find its solution in the domain III, determining $y_l^c = S(-\pi)$, which is used for the conservative lock-in frequency $\omega_l^c$ computation (see \eqref{eq:omega_l^c_formula-appendix}).
		Parameters:
		$\tau_1 = 0.0633$,
		$\tau_2 = 0.0225$,
		$K_{\rm vco} = 250$,
		$k = \frac{2}{\pi}$.
	}
		\label{fig:separatrix integration}
	\end{figure}
	
	\textbf{Domain I}.
	
	The saddle separatrix is locally described by the saddle's eigenvectors
	\begin{equation*}
	\begin{aligned}
	&
	V^s_+
	=
	\begin{pmatrix}
	1\\
	\frac{c - a}{2}
	\end{pmatrix},
	\quad
	V^s_-
	=
	\begin{pmatrix}
	1\\
	\frac{-c - a}{2}.
	\end{pmatrix}
	\end{aligned}
	\end{equation*}
	Eigenvector $V^s_-$ points to a saddle and $V^s_+$ has the opposite direction.
	Since in the considered domain the system is a linear one, then the separatrix coincides with the line corresponding to $V_-^s$:
	\begin{equation*}
	\begin{aligned}
	&S(\theta_e) = \frac{c - a}{2 \left(\pi - \frac{1}{k}\right)} (\pi - \theta_e),
	\quad \frac{1}{k} < \theta_e < \pi.
	\end{aligned}
	\end{equation*}
	Let's obtain the limit value in $\theta_e = \frac{1}{k}$:
	\begin{equation*}
	\begin{aligned}
	& S(\frac{1}{k}) = \frac{c - a}{2}>0.
	\end{aligned}
	\end{equation*}
	
	\textbf{Domain II}.
	If $-\frac{1}{k} < \theta_e(t) < \frac{1}{k}$ then system \eqref{eq:PLL-triangular-after_change_of_variables} is
	\begin{equation}\label{eq:linear equation in interval II}
	\begin{aligned}
	&\dot{y} = -a k y - k \theta_e,\\
	&\dot\theta_e = y.
	\end{aligned}
	\end{equation}
	In the domains $\{y>0\}$ and $\{y<0\}$, variable $\theta_e(t)$ changes monotonically and the behaviour of system \eqref{eq:linear equation in interval II} can be described by the first-order differential equation\footnote{
		The similar transition to the first-order differential equation was used in \cite{Belyustina-1959, HuqueS-2011-exact, HuqueS-2013}.
	}:
	\begin{equation}\label{eq:Chini equation in interval II}
	\begin{aligned}
	& \frac{dy}{d\theta_e} = -ak - \frac{k\theta_e}{y}.
	\end{aligned}
	\end{equation}
	The obtained equation is Chini's equation \cite{Chini-1924, ChebK-2003}, which is a generalization of Abel and Riccati equations.
	The change of variables $z = \frac{y}{\theta_e}$ maps equation \eqref{eq:Chini equation in interval II} into a separable one\footnote{
		The same change of variables was used in \cite{HuqueS-2011-exact, HuqueS-2013}.
	}:
	\begin{equation}\label{eq:PLL-equation-to_integrate interval II}
	\begin{aligned}
	& \frac{z dz}{z^2 + akz + k} = -\frac{d\theta_e}{\theta_e}.
	\end{aligned}
	\end{equation}
	If $z^2+akz+k\ne0$ then solutions of system \eqref{eq:Chini equation in interval II} and system \eqref{eq:PLL-equation-to_integrate interval II} coincide in domains $0 < \theta_e < \frac{1}{k}$ and $-\frac{1}{k} < \theta_e < 0$.
	Depending on the type of an asymptotically stable equilibrium, the following cases appear (see section~\ref{sec:local analysis}):
	\begin{itemize}
		\item $a^2k > 4$ (the equation $z^2+akz+k = 0$ describes the eigenvectors of the stable node),
		\item $a^2k = 4$ (the equation $z^2+akz+k = 0$ describes the eigenvector of the stable degenerate node),
		\item $a^2k < 4$ (here the case $z^2+akz+k = 0$ is not possible).
	\end{itemize}
	
	\textbf{Case $a^2k > 4$}.
	Let's take into account the location of separatrix $y = S(\theta_e)$, satisfying \eqref{eq:Chini equation in interval II}, during its integration on intervals.
	The eigenvectors of the stable node
	\begin{equation*}
	\begin{aligned}
	&
	V^n_+
	=
	\begin{pmatrix}
	1\\
	-\frac{a - b}{2}
	\end{pmatrix},
	\quad
	V^n_-
	=
	\begin{pmatrix}
	1\\
	-\frac{a + b}{2}
	\end{pmatrix}
	\end{aligned}
	\end{equation*}
	are described by lines	$y = - \frac{a+b}{2}k\theta_e$ and $y = - \frac{a-b}{2}k\theta_e$, respectively, and intersect the boundary $\theta_e = \frac{1}{k}$ of domains I and II in points $-\frac{a - b}{2} < 0$ and $-\frac{a + b}{2} < 0$.
	Hence, the separatrix, intersecting the boundary $\theta_e = \frac{1}{k}$ of domains I and II in point $\frac{c - a}{2} > 0$, remains over the eigenvectors within the domain II and satisfies the following inequality: $(S(\theta_e) + \frac{a + b}{2}k\theta_e)(S(\theta_e) + \frac{a - b}{2}k\theta_e) > 0$ as $\theta_e\in[-\frac{1}{k},\frac{1}{k}]$.
	
	Assuming $(z + \frac{a + b}{2}k)(z + \frac{a - b}{2}k) > 0$, the general solution of equation \eqref{eq:PLL-equation-to_integrate interval II} is as follows \footnote{
		Taking derivative of $N_1(z)$, we have
		\begin{equation*}
		\begin{aligned}
		& N_1^\prime(z)
		=
		\frac{1}{2}
		\frac{1}{
			\big(z + \frac{a - b}{2}k\big)^{\frac{b - a}{b}}
			\big(z + \frac{a + b}{2}k\big)^{\frac{b + a}{b}}
		}
		\Big(
		\frac{b-a}{b}
		\big(z + \frac{a - b}{2}k\big)^{\frac{ - a}{b}}
		\big(z + \frac{a + b}{2}k\big)^{\frac{b + a}{b}}
		+\\
		&+
		\frac{b+a}{b}
		\big(z + \frac{a - b}{2}k\big)^{\frac{b - a}{b}}
		\big(z + \frac{a + b}{2}k\big)^{\frac{a}{b}}
		\Big)
		=
		\frac{1}{2}
		\Big(
		\frac{b-a}{b}
		\big(z + \frac{a - b}{2}k\big)^{-1}
		+
		\frac{b+a}{b}
		\big(z + \frac{a + b}{2}k\big)^{-1}
		\Big)
		=\\
		&=
		\frac{1}{2(z + \frac{a - b}{2}k)(z + \frac{a + b}{2}k)}
		\Big(
		\frac{b-a}{b}
		\big(z + \frac{a + b}{2}k\big)
		+
		\frac{b+a}{b}
		\big(z + \frac{a - b}{2}k\big)
		\Big)
		=
		\frac{z}{(z + \frac{a - b}{2}k)(z + \frac{a + b}{2}k)}
		=\\
		&=
		\frac{z}{z^2 + akz + k}.
		\end{aligned}
		\end{equation*}
	}:
	\begin{equation*}
	\begin{aligned}
	& N_1(z) = -\ln|\theta_e| + C
	\end{aligned}
	\end{equation*}
	where
	\begin{equation*}
	\begin{aligned}
	& N_1(z) 
	=
	\frac{1}{2}
	\ln
	\Big(
	\big(z + \frac{a - b}{2}k\big)^{\frac{b - a}{b}}
	\big(z + \frac{a + b}{2}k\big)^{\frac{b + a}{b}}
	\Big),
	\quad
	C = {\rm const}.
	\end{aligned}
	\end{equation*}
	Since for separatrix $y = S(\theta_e)$ inequality $(y + \frac{a + b}{2}k\theta_e)(y + \frac{a - b}{2}k\theta_e) > 0$ is valid, we get that the separatrix on interval $0<\theta_e(t) \le \frac{1}{k}$ satisfies $N(y, \theta_e) = C_{(0, \frac{1}{k})}$ where
	\begin{equation*}
	\begin{aligned}
	& N(y, \theta_e) 
	=
	\frac{1}{2}
	\ln
	\Big(
	\big(y + \frac{a - b}{2}k\theta_e\big)^{\frac{b - a}{b}}
	\big(y + \frac{a + b}{2}k\theta_e\big)^{\frac{b + a}{b}}
	\Big),\\
	&C_{(0, \frac{1}{k})} = 
		\lim\limits_{\theta_e\to\frac{1}{k}-0} N(\frac{c-a}{2}, \theta_e)
		=
		\frac{1}{2}
		\ln
		\Big(
		(\frac{c - b}{2})^{\frac{b - a}{b}}
		(\frac{c + b}{2})^{\frac{b + a}{b}}
		\Big)
		=
		\frac{1}{2}
		\ln
		\Big(
		\pi
		(\frac{c + b}{c - b})^{\frac{a}{b}}
		\Big).
	\end{aligned}
\end{equation*}
	Thus, if $a^2k>4$, then separatrix $y = S(\theta_e)$ in domain $0 < \theta_e(t) \le \frac{1}{k}$ is described by equation
	\begin{equation}\label{eq:separatrix II case 1}
	\begin{aligned}
	& 
	\big(y + \frac{a - b}{2}k\theta_e\big)^{\frac{b - a}{b}}
	\big(y + \frac{a + b}{2}k\theta_e\big)^{\frac{b + a}{b}}
	=
	\pi
	(\frac{c + b}{c - b})^{\frac{a}{b}}.
	\end{aligned}
	\end{equation}
	Substituting $\theta_e \to +0$ into \eqref{eq:separatrix II case 1}, we get
	\begin{equation}\label{eq:equation-for-y2-case1}
	\begin{aligned}
	&y_l = \sqrt{\pi}
	(\frac{c + b}{c - b})^{\frac{a}{2b}}.
	\end{aligned}
	\end{equation}
	Then, substituting \eqref{eq:equation-for-y2-case1} into \eqref{eq:omega_l_formula-appendix}, we get the first case of formula \eqref{eq:lock-in_exact_formula}.
	
	To determine the conservative lock-in frequency, we firstly need to get $d = S(-\frac{1}{k})$, then to obtain the equation for the separatrix in domain III, and, finally, to determine $y_l^c = S(-\pi)$.
	Since the separatrix on interval $-\frac{1}{k} < \theta_e(t) < 0$ satisfies $N(y, \theta_e) = C_{(-\frac{1}{k}, 0)}$ and 	$\lim\limits_{\theta_e\to+0} N(y, \theta_e)  = 
	\lim\limits_{\theta_e\to-0} N(y, \theta_e) 
	=
	\ln y$ as $y>0$, then $C_{(-\frac{1}{k}, 0)} = C_{(0, \frac{1}{k})} = \frac{1}{2}
	\ln
	\Big(
	\pi
	(\frac{c + b}{c - b})^{\frac{a}{b}}
	\Big)$.	
	
	Thus, if $a^2k > 4$, then separatrix $y = S(\theta_e)$ in domain II is described by equation \eqref{eq:separatrix II case 1}.
	Substituting $\theta_e = -\frac{1}{k}$ into \eqref{eq:separatrix II case 1}, we get
	\begin{equation}\label{eq:equation-for-y23-case1}
	\begin{aligned}
	& 
	(d - \frac{a - b}{2})^{\frac{b - a}{b}}
	(d - \frac{a + b}{2})^{\frac{b + a}{b}}
	=
	\pi
	(\frac{c + b}{c - b})^{\frac{a}{b}}.
	\end{aligned}
	\end{equation}
	Since the separatrix is over the eigenvectors ($y > - \frac{a\pm b}{2}k\theta_e$), then
	\begin{equation*}
	\begin{aligned}
	& d
	>
	\frac{a + b}{2}.
	\end{aligned}
	\end{equation*}
	Notice that if $d = \frac{a + b}{2}$, then the left-hand side of equation \eqref{eq:equation-for-y23-case1} equals to zero, but the right-hand side is positive.
	Then the left-hand side increases monotonically as value $d$ increases
	and tends to infinity as $d \to+\infty$.
	Thus, equation \eqref{eq:equation-for-y23-case1} has unique solution $d$ greater than $\frac{a + b}{2}$.
	
	\textbf{Case $a^2k = 4$}.
	
	In domain II, separatrix $y = S(\theta_e)$ is over eigenvector
	\begin{equation*}
		\begin{aligned}
			&
			V^{dn}
			=
			\begin{pmatrix}
				1\\
				-\frac{a}{2}
			\end{pmatrix},
		\end{aligned}
	\end{equation*}
	which is described by line $y = -\frac{2}{a}\theta_e$ and intersects the boundary $\theta_e = \frac{1}{k}$ of domains I and II in point $ - \frac{a}{2} < 0$.
	Hence, the separatrix, intersecting the boundary $\theta_e = \frac{1}{k}$ of domains I and II in point $\frac{c - a}{2} > 0$, remains over the eigenvector within the domain II and satisfies the following inequality: $S(\theta_e) > -\frac{2}{a}\theta_e$.
	
	The general solution of \eqref{eq:PLL-equation-to_integrate interval II} is as follows\footnote{
		Taking derivative of $N_1(z)$, we have
		\begin{equation*}
		\begin{aligned}
		& N_1^\prime(z)
		=
		- \frac{2a}{(2+az)^2}
		+
		\frac{a}{2 + az}
		=
		\frac{a^2z}{(2+az)^2}
		=
		\frac{z}{(\frac{2}{a} + z)^2}
		=
		\frac{z}{z^2 + akz + k}.
		\end{aligned}
		\end{equation*}
	}:
	\begin{equation*}
	\begin{aligned}
	& N_1(z) = -\ln|\theta_e| + C
	\end{aligned}
	\end{equation*}
	where
	\begin{equation*}
	\begin{aligned}
	& N_1(z) = \frac{2}{2+az} + \ln|2+az|,
	\quad 
	C= {\rm const}.
	\end{aligned}
	\end{equation*}
	
	Since for separatrix $y = S(\theta_e)$ inequality $S(\theta_e) > -\frac{2}{a}\theta_e$. is valid, we get that the separatrix on interval $0 < \theta_e(t) \le \frac{1}{k}$ satisfies $N(y, \theta_e) = C_{(0, \frac{1}{k})}$ where
	\begin{equation*}
	\begin{aligned}
	& N(y, \theta_e) 
	=
	\frac{2\theta_e}{2\theta_e+ay} + \ln(2\theta_e + ay),\\
	&C_{(0, \frac{1}{k})} = 
	\lim\limits_{\theta_e\to\frac{1}{k}-0} N(\frac{c-a}{2}, \theta_e)
	=
	\frac{2}{k}\frac{1}{\frac{2}{k} + a\sqrt{\pi} - \frac{a^2}{2}} + \ln(\frac{2}{k} + a\sqrt{\pi} - \frac{a^2}{2})
	=
	\frac{a}{2\sqrt{\pi}} + \ln(a\sqrt{\pi}).
	\end{aligned}
	\end{equation*}
	Thus, if $a^2k = 4$, then separatrix $y = S(\theta_e)$ in domain $0 < \theta_e(t) \le \frac{1}{k}$ is described by equation
	\begin{equation}\label{eq:separatrix II case 2}
	\begin{aligned}
	& 
	\frac{2\theta_e}{2\theta_e+ay} + \ln(2\theta_e + ay)
	=
	\frac{a}{2\sqrt{\pi}} + \ln(a\sqrt{\pi}).
	\end{aligned}
	\end{equation}
	Substituting $\theta_e \to +0$ into \eqref{eq:separatrix II case 2}, we get
	\begin{equation}\label{eq:equation-for-y2-case2}
	\begin{aligned}
	&y_l = \sqrt{\pi}\exp(\frac{a}{2\sqrt{\pi}}).
	\end{aligned}
	\end{equation}
	Then, substituting \eqref{eq:equation-for-y2-case2} into \eqref{eq:omega_l_formula-appendix}, we get the second case of formula \eqref{eq:lock-in_exact_formula}.
	
	To determine the conservative lock-in frequency, we firstly need to determine $d = S(-\frac{1}{k})$.
	Since the separatrix on interval $-\frac{1}{k} < \theta_e(t) < 0$ satisfies $N(y, \theta_e) = C_{(-\frac{1}{k}, 0)}$ and $\lim\limits_{\theta_e\to+0} N(y, \theta_e)  = 
	\lim\limits_{\theta_e\to-0} N(y, \theta_e) 
	=
	\ln(ay)$ as $y>0$,
	then $C_{(-\frac{1}{k}, 0)} = C_{(0, \frac{1}{k})} = \frac{a}{2\sqrt{\pi}} + \ln(a\sqrt{\pi})$.	
	
	Thus, if $a^2k = 4$, then separatrix $y = S(\theta_e)$ in domain II is described by equation \eqref{eq:separatrix II case 2}.	
	Substituting $\theta_e = -\frac{1}{k}$ into \eqref{eq:separatrix II case 2}, we get
	\begin{equation*}
	\begin{aligned}
	& (d - \frac{a}{2})
	\exp\Big(
	\frac{\frac{a}{2}}{\frac{a}{2} - d}
	\Big)
	= 
	\sqrt{\pi}e^{\frac{a}{2\sqrt{\pi}}}.
	\end{aligned}
	\end{equation*}
	Notice that in the considered case it is possible to obtain an explicit formula for $d$:	
	\begin{equation}\label{eq:equation-for-y23-case2}
	\begin{aligned}
	& d 
	= 
	\frac{a}{2}
	\Big(1 + \frac{1}{W(\frac{a}{2\sqrt{\pi}} 
		\exp(-\frac{a}{2\sqrt{\pi}})}\Big)
	\end{aligned}
	\end{equation}
	where $W(x)$ is the Lambert W function\footnote{
		For $x>0$ function $W(x)$ is a single-valued one and can be evaluated in standard numeric computing platforms.}.
	
	\textbf{Case $a^2k < 4$}.
	
	The general solution of \eqref{eq:Chini equation in interval II} is as follows\footnote{
		Taking derivative of $N_1(z)$, we have
		\begin{equation*}
		\begin{aligned}
		& N_1^\prime(z)
		=
		\frac{2z+ak}{2(z^2 + akz + k)}
		-
		\frac{2a}{b^2k}\frac{1}{1 + (\frac{a + \frac{2}{k}z}{b})^2}
		=
		\frac{2z+ak}{2(z^2 + akz + k)}
		-
		\frac{2ak}{b^2k^2 + (ak + 2z)^2}
		=\\
		&=
		\frac{2z+ak}{2(z^2 + akz + k)}
		-
		\frac{2ak}{4z^2 + 4akz + (a^2+b^2)k^2}
		=
		\frac{2z+ak}{2(z^2 + akz + k)}
		-
		\frac{ak}{2(z^2 + akz + k)}
		=
		\frac{z}{z^2 + akz + k}.
		\end{aligned}
		\end{equation*}
	}:
	\begin{equation*}
	\begin{aligned}
	& N_1(z) = -\ln|\theta_e| + C
	\end{aligned}
	\end{equation*}
	where 
	\begin{equation*}
	\begin{aligned}
	& N_1(z) = \frac{1}{2}\ln(z^2 + akz + k) - \frac{a}{b} \arctan(\frac{a + \frac{2}{k}z}{b}),
	\quad
	C= {\rm const}.
	\end{aligned}
	\end{equation*}
	Then separatrix $y = S(\theta_e)$ in domain $0 < \theta_e \le \frac{1}{k}$ satisfies $N(y, \theta_e) = C_{(0,\ \frac{1}{k})}$ where
	\begin{equation*}
	\begin{aligned}
	& N(y, \theta_e) 
	=
	\frac{1}{2}\ln(y^2 + aky\theta_e + k \theta^2_e)
	- 
	\frac{a}{b} 
	\arctan\Big(\frac{a\theta_e + \frac{2}{k}y}{b\theta_e}\Big),\\
	&C_{(0,\ \frac{1}{k})} = 
	\lim\limits_{\theta_e\to\frac{1}{k}-0} N(\frac{c-a}{2}, \theta_e)
	=
	\frac{1}{2}\ln\pi
	- \frac{a}{b} 
	\arctan\frac{c}{b}.
	\end{aligned}
	\end{equation*}
	Thus, if $a^2k < 4$, then separatrix $y = S(\theta_e)$ in domain $0 < \theta_e \le \frac{1}{k}$ is described by equation
	\begin{equation}\label{eq:separatrix II case 3-1}
	\begin{aligned}
	& 
	\frac{1}{2}\ln(y^2 + aky\theta_e + k \theta^2_e)
	- 
	\frac{a}{b} 
	\arctan\Big(\frac{a\theta_e + \frac{2}{k}y}{b\theta_e}\Big)
	=
	\frac{1}{2}\ln\pi
	- \frac{a}{b} 
	\arctan\frac{c}{b}.
	\end{aligned}
	\end{equation}
	Substituting $\theta_e \to +0$ into \eqref{eq:separatrix II case 3-1}, we get
	\begin{equation}\label{eq:equation-for-y2-case3}
	\begin{aligned}
	&y_l = \sqrt{\pi}
	\exp\Big({\frac{a}{b}
		\arctan\frac{b}{c}}\Big).
	\end{aligned}
	\end{equation}
	Then, substituting \eqref{eq:equation-for-y2-case3} into \eqref{eq:omega_l_formula-appendix}, we get the third case of formula \eqref{eq:lock-in_exact_formula}.
	Thus, Theorem~\ref{theorem: lock-in stable} is proved.

	To determine the conservative lock-in frequency, we firstly need to determine $d = S(-\frac{1}{k})$.
	Since the separatrix on interval $-\frac{1}{k} < \theta_e(t) < 0$ satisfies $N(y, \theta_e) = C_{(-\frac{1}{k}, 0)}$ and $\lim\limits_{\theta_e\to+0} N(y, \theta_e)  
	=
	\ln y - \frac{\pi a}{2b}$, 
	$\lim\limits_{\theta_e\to-0} N(y, \theta_e) 
	=
	\ln y + \frac{\pi a}{2b}$,
	then $C_{(-\frac{1}{k}, 0)} - \frac{\pi a}{2b}= C_{(0, \frac{1}{k})} + \frac{\pi a}{2b}$.

	Thus, if $a^2k < 4$, then separatrix $y = S(\theta_e)$ in domain II is described by
	\begin{equation}\label{eq:separatrix II case 3}
	\begin{aligned}
	& 
	\frac{1}{2}\ln(y^2 + aky\theta_e + k \theta^2_e)
	- 
	\frac{a}{b} 
	\arctan\Big(\frac{a\theta_e + \frac{2}{k}y}{b\theta_e}\Big)
	=
	\frac{1}{2}\ln\pi
	- \frac{a}{b} 
	\arctan\frac{c}{b}, \quad \text{if} \; 0 < \theta_e(t) \le \frac{1}{k}\\
	&
	y = y_l, \quad \text{if} \; \theta_e(t) = 0\\
	& 
	\frac{1}{2}\ln(y^2 + aky\theta_e + k \theta^2_e)
	- 
	\frac{a}{b} 
	\arctan\Big(\frac{a\theta_e + \frac{2}{k}y}{b\theta_e}\Big)
	=
	\frac{1}{2}\ln\pi
	+ \frac{a}{b}
	\Big(
	\pi
	-
	\arctan\frac{c}{b}
	\Big), \quad \text{if} \; -\frac{1}{k} \le \theta_e(t) < 0.
	\end{aligned}
	\end{equation}
	Substituting $\theta_e = -\frac{1}{k}$ into \eqref{eq:separatrix II case 3}, we get
	\begin{equation}\label{eq:equation-for-y23-case3-first}
	\begin{aligned}
	& 
	\Big(
	d^2 - ad + \frac{1}{k}
	\Big)
	\exp
	\Big(
	\frac{2a}{b} 
	\arctan\frac{2d - a}{b}
	-
	\frac{\pi a}{b}
	\Big)
	=
	\pi
	\exp
	\Big(\frac{2a}{b}
	\arctan\frac{b}{c}
	\Big).
	\end{aligned}
	\end{equation}	
	Notice that if $d = 0$, then the left-hand side of equation \eqref{eq:equation-for-y23-case1} is less than the right-hand side:
	\begin{equation*}
	\begin{aligned}
	&
	\frac{1}{k}
	\exp
	\Big(
	\frac{2a}{b} 
	\arctan\frac{- a}{b}
	-
	\frac{\pi a}{b}
	\Big)
	<
	\frac{1}{k} 
	<
	\pi
	<
	\pi
	\exp\Big(
	\frac{2a}{b}
	\arctan\frac{b}{c}
	\Big).
	\end{aligned}
	\end{equation*}
	Then the left-hand side increases monotonically as value $d$ increases and tends to infinity as $d \to+\infty$.
	Thus, equation \eqref{eq:equation-for-y23-case3-first} has unique positive solution $d$.
	
	Notice also that if $d = \frac{a}{2}$, then the left-hand side of equation \eqref{eq:equation-for-y23-case1} is less than the right-hand side too:
	\begin{equation*}
	\begin{aligned}
	& 
	b^2
	\exp(-\frac{\pi a}{b})
	<
	a^2
	\exp(-\frac{\pi a}{b})
	< 4\pi \exp(-\frac{\pi a}{b})
	< 4\pi \frac{b^2}{\pi^2 a^2}
	< \frac{4}{\pi} < \pi 
	<
	\pi	\exp\Big(\frac{2a}{b}\arctan\frac{b}{c}\Big).
	\end{aligned}
	\end{equation*}
	Thus, $d > \frac{a}{2}$ and equation \eqref{eq:equation-for-y23-case3-first} can be reduced to the following: \begin{equation}\label{eq:equation-for-y23-case3}
	\begin{aligned}
	& 
	\Big(
	d^2 - ad + \frac{1}{k}
	\Big)
	\exp
	\Big(
	\frac{2a}{b} \arctan\frac{b}{a - 2d}
	\Big)
	=
	\pi
	\exp
	\Big(\frac{2a}{b}
	\arctan\frac{b}{c}
	\Big).
	\end{aligned}
	\end{equation}

	\textbf{Domain III}
	
	If $ -\pi \le \theta_e(t) < -\frac{1}{k}$ then system \eqref{eq:PLL-triangular-after_change_of_variables} is
	\begin{equation}\label{eq:linear equation in interval III}
	\begin{aligned}
	&\dot{y} = \frac{a}{\pi - \frac{1}{k}} y + \frac{1}{\pi - \frac{1}{k}} (\theta_e + \pi),\\
	& \dot{\theta_e} = y.
	\end{aligned}
	\end{equation}
	Analogously to the analysis in domain II, let's study for $y>0$ the first-order differential equation
	\begin{equation}\label{eq:Chini equation in interval III}
	\begin{aligned}
	& \frac{dy}{d\theta_e} = \frac{1}{\pi - \frac{1}{k}}(a + \frac{\theta_e + \pi}{y})
	\end{aligned}
	\end{equation}
	and make the change of variables $z = \frac{y}{\theta_e+\pi}$ mapping equation \eqref{eq:Chini equation in interval III} into a separable one:
	\begin{equation}\label{eq:PLL-equation-to_integrate interval III}
	\begin{aligned}
	&(\pi - \frac{1}{k})\frac{zdz}{\left(\pi - \frac{1}{k}\right)z^2 - az - 1} = - \frac{d\theta_e}{\theta_e + \pi}.
	\end{aligned}
	\end{equation}
	If $-\pi < \theta_e(t) < -\frac{1}{k}$ then the solutions of system \eqref{eq:Chini equation in interval III} and system \eqref{eq:PLL-equation-to_integrate interval III} coincide.
	
	Separatrix $y = S(\theta_e)$ is over the separatrices of saddle $(0,\ -\pi)$, which are described by the equations
	\begin{equation*}
	\begin{aligned}
	&y = \frac{\pm c - a}{2 (\pi - \frac{1}{k})} (-\pi - \theta_e),
	\quad -\pi < \theta_e < -\frac{1}{k}.
	\end{aligned}
	\end{equation*}
	Thus, the following inequality is valid for the separatrix: 
	$(S(\theta_e) + \frac{c - a}{2(\pi - \frac{1}{k})}(\pi + \theta_e))
	(S(\theta_e) - \frac{c + a}{2(\pi - \frac{1}{k})}(\pi + \theta_e)) > 0$.
	
	Assuming $(z + \frac{c - a}{2(\pi - \frac{1}{k})})(z - \frac{c + a}{2(\pi - \frac{1}{k})}) > 0$, the general solution of equation \eqref{eq:PLL-equation-to_integrate interval III} is as follows\footnote{
		Taking derivative of $M_1(z)$, we have
		\begin{equation*}
		\begin{aligned}
		& M_1^\prime(z)
		=
		\frac{1}{2}
		\frac{1}{
			\big(
			z + \frac{c - a}{2(\pi - \frac{1}{k})}
			\big)^{\frac{c - a}{c}}
			\big(
			z - \frac{c + a}{2(\pi - \frac{1}{k})}
			\big)^{\frac{c + a}{c}}
		}
		\Big(
		\frac{c-a}{c}
		\big(z + \frac{c - a}{2(\pi - \frac{1}{k})}\big)^{\frac{- a}{c}}
		\big(z - \frac{c + a}{2(\pi - \frac{1}{k})}\big)^{\frac{c + a}{c}}
		+\\
		&+
		\frac{c+a}{c}
		\big(z + \frac{c - a}{2(\pi - \frac{1}{k})}\big)^{\frac{c - a}{c}}
		\big(z - \frac{c + a}{2(\pi - \frac{1}{k})}\big)^{\frac{a}{c}}
		\Big)
		=
		\frac{1}{2}
		\Big(
		\frac{c-a}{c}
		\big(z + \frac{c - a}{2(\pi - \frac{1}{k})}\big)^{-1}
		+
		\frac{c+a}{c}
		\big(z - \frac{c + a}{2(\pi - \frac{1}{k})}\big)^{-1}
		\Big)
		=\\
		&=
		\frac{1}{2
			(z + \frac{c - a}{2(\pi - \frac{1}{k})})
			(z - \frac{c + a}{2(\pi - \frac{1}{k})})}
		\Big(
		\frac{c-a}{c}
		\big(z - \frac{c + a}{2(\pi - \frac{1}{k})}\big)
		+
		\frac{c+a}{c}
		\big(z + \frac{c - a}{2(\pi - \frac{1}{k})}\big)
		\Big)
		=
		\frac{z}{
			(z + \frac{c - a}{2(\pi - \frac{1}{k})})
			(z - \frac{c + a}{2(\pi - \frac{1}{k})})}
		=\\
		&=
		\frac{z}{
			(z^2 - \frac{a}{\pi - \frac{1}{k}}z - \frac{1}{\pi - \frac{1}{k}})}.
		\end{aligned}
		\end{equation*}
	}:
	\begin{equation*}
	\begin{aligned}
	& M_1(z) = -\ln|\theta_e+\pi| + C
	\end{aligned}
	\end{equation*}
	where 
	\begin{equation*}
	\begin{aligned}
	& M_1(z) 
	=
	\frac{1}{2}
	\ln 
	\Big(
	\big(
	z + \frac{c - a}{2(\pi - \frac{1}{k})}
	\big)^{\frac{c - a}{c}}
	\big(
	z - \frac{c + a}{2(\pi - \frac{1}{k})}
	\big)^{\frac{c + a}{c}}
	\Big),\\
	&C = {\rm const}.
	\end{aligned}
	\end{equation*}
	Since for separatrix $y = S(\theta_e)$ inequality $(y + \frac{c - a}{2(\pi - \frac{1}{k})}(\pi + \theta_e))(y - \frac{c + a}{2(\pi - \frac{1}{k})}(\pi + \theta_e)) > 0$ is valid, we get that the separatrix in domain III satisfies $M(y, \theta_e) = C_{\left(-\pi, -\frac{1}{k}\right)}$ where
	\begin{equation*}
	\begin{aligned}
	&M(y, \theta_e) 
	=
	\frac{1}{2}
	\ln 
	\Big(
	\Big(
	y + \frac{c - a}{2(\pi - \frac{1}{k})}(\pi + \theta_e)
	\Big)^{\frac{c - a}{c}}
	\Big(
	y - \frac{c + a}{2(\pi - \frac{1}{k})}(\pi + \theta_e)
	\Big)^{\frac{c + a}{c}}
	\Big),\\
	&C = C_{\left(-\pi, -\frac{1}{k}\right)} = \lim_{\theta_e \to -\frac{1}{k}-0}
	M(d, \theta_e)
	= 
	\frac{1}{2}
	\ln 
	\Big(
	(d + \frac{c - a}{2})^{\frac{c - a}{c}}
	(d - \frac{c + a}{2})^{\frac{c + a}{c}}
	\Big).
	\end{aligned}
	\end{equation*}
	Thus, separatrix $y = S(\theta_e)$ in domain III is described by equation
	\begin{equation}\label{eq:separatrix III}
	\begin{aligned}
	& 
	\Big(
	y + \frac{c - a}{2(\pi - \frac{1}{k})}(\pi + \theta_e)
	\Big)^{\frac{c - a}{c}}
	\Big(
	y - \frac{c + a}{2(\pi - \frac{1}{k})}(\pi + \theta_e)
	\Big)^{\frac{c + a}{c}}
	=
	(d + \frac{c - a}{2})^{\frac{c - a}{c}}
	(d - \frac{c + a}{2})^{\frac{c + a}{c}}.
	\end{aligned}
	\end{equation}
	To determine the conservative lock-in frequency, we firstly need to determine $y_l^c = S(-\pi)$.
	Substituting $\theta_e = - \pi$ into \eqref{eq:separatrix III}, we get
	\begin{equation}\label{eq:formula for y^III}
	\begin{aligned}
	&y_l^c
	=
	(d + \frac{c - a}{2})^{\frac{c - a}{2c}}
	(d - \frac{c + a}{2})^{\frac{c + a}{2c}}.
	\end{aligned}
	\end{equation}
	
	Substituting \eqref{eq:formula for y^III} into \eqref{eq:omega_l^c_formula-appendix} and taking into account formulae \eqref{eq:equation-for-y23-case1}, \eqref{eq:equation-for-y23-case2}, \eqref{eq:equation-for-y23-case3}, we get \eqref{eq:conservative_lock-in_exact_formula}.
	
	Theorem~\ref{theorem: lock-in stable} and Theorem~\ref{theorem: lock-in unstable} are proved.	
\end{proof}

\appendix{Octave code for Fig.~\ref{fig:separatrix integration}}
Code below can be runned on https://octave-online.net/ in order to obtain phase portrait on Fig.~\ref{fig:separatrix integration} and verify formulae \eqref{eq:lock-in_exact_formula} and \eqref{eq:conservative_lock-in_exact_formula}.
The code simulates trajectories of system \eqref{eq:PLL-triangular-after_change_of_variables} numerically and additionally plots two points: $(0, y_l)$ and $(-\pi, y_l^c)$ where $y_l$ and $y_l^c$ are used in lock-in range formulae \eqref{eq:omega_l_formula-appendix} and \eqref{eq:omega_l^c_formula-appendix}.
Since these points are lying on the separatrices, formulae \eqref{eq:omega_l_formula-appendix} and \eqref{eq:omega_l^c_formula-appendix} are validated numerically.

\begin{lstlisting}
clear all;
close all;
clc;

% VCO input gain
K_vco = 250;

% Loop filter transfer function F(s) = (1+s tau_2)/(s tau_1), tau_1 > 0,
% tau_2 > 0
tau_2 = 0.0225;
tau_1 = 0.0633;

% The slope coefficient of piecewise-linear PD characteristic
% it becomes triangular with this coefficient
k=2/pi;

function y = draw_saddles_symmetric(ode, saddle, V, periods, options)
% draw_saddles_symmetric draws phase portrait of the system ode, 
% starting from saddle with eigenvectors V for given number of periods
% options are used to tune ODE solver

    % Integration time
    TIME_POSITIVE = 0:0.005:20;
    TIME_NEGATIVE = 0:-0.005:-20;
   
    v1 = 0.01 .* flip(V(:,1)');
    v2 = 0.01 .* flip(V(:,2)');
    
    % Calculate saddle separatrices and plot them
    [T_1, X_1] = ode45(ode, TIME_NEGATIVE, saddle+v1, options);
    [T_3, X_3] = ode45(ode, TIME_NEGATIVE, saddle-v1, options);
    for j=1:length(periods)
        plot(X_1(:,2)+periods(j),X_1(:,1));
        plot(X_3(:,2)+periods(j),X_3(:,1));
    end
    [T_2, X_2] = ode45(ode, TIME_POSITIVE, saddle+v2, options);
    [T_4, X_4] = ode45(ode, TIME_POSITIVE, saddle-v2, options);
    for j=1:length(periods)
       plot(X_2(:,2)+periods(j), X_2(:,1));
       plot(X_4(:,2)+periods(j), X_4(:,1));
    end
end

function y = sawtooth_diff(t)
% sawtooth_diff - the derivative of a triangular function 
% (sawtooth (T, WIDTH) from signal package)
  remain = abs(rem(t, 2*pi));
  if (pi/2 <= remain && remain <= 3*pi/2)
      y = -2/pi;
  else
      y = 2/pi;
  end
end

% PD characteristic
v_e = @(theta_e) (sawtooth(theta_e+pi/2,0.5));
% The derivative of PD characteristic
dv_e = @(theta_e) (sawtooth_diff(theta_e));
period = 2*pi;

% Parameters a, b, c from Theorem 1 and 2
a = sqrt(K_vco/tau_1)*tau_2;
b = sqrt(abs(a^2-4/k));
c = sqrt(a^2-4/k + 4*pi);
% Computing the lock-in range and the concervative lock-in range by 
% Theorem 1 and 2
syms x;
if a^2*k>4
    y_l = sqrt(pi)* ((c+b)/(c-b))^(a/(2*b));
    fcn = @(x) (x - (a-b)/2)^((b-a)/b) * (x - (a+b)/2)^((b+a)/b) -...
        pi*((c+b)/(c-b))^(a/b);
    init_param = [(a+b)/2, 10000000];
    % vpasolve numerically solves implicit equations with initial guess
    % init_param (it was proven that the equation has a unique solution
    % for x > (a+b)/2)
    d = vpasolve(fcn(x), x, init_param);
else 
    if a^2*k == 4
        y_l = sqrt(pi)*exp(a/2/sqrt(pi));
        d = a/2*(1 + 1/lambertw(a/2/sqrt(pi)*exp(-a/2/sqrt(pi))));
else
        y_l = sqrt(pi)*exp(a*atan(b/c)/b);
        fcn = @(x) ((x)^2 - a*x + 1/k) * exp( 2*a/b*atan((2*x-a)/b) - pi*a/b) -...
        pi*exp(2*a/b*atan(b/c));
        init_param = [a/2, 10000000];
        % vpasolve numerically solves implicit equations with initial guess
        % init_param (it was proven that the equation has a unique solution
        % for x > a/2)
        d = vpasolve(fcn(x), x, init_param);
    end
end
y_l_c = (d-0.5*(a-c))^((c-a)/c/2) * (d-0.5*(a+c))^((c+a)/c/2);


h = figure(1);
hold on;
grid on;

% Ploting two points which correspond the lock-in range and the conservative
% lock-in range
plot([0], [y_l], 'k.', 'MarkerSize', 20);
plot([-pi], [eval(y_l_c)], 'k.', 'MarkerSize', 20);

% System establishment for numerocal integration
% y = x(1), theta_e = x(2)
pll_s = @(t,x)([- a*dv_e(x(2))*x(1) - v_e(x(2));
                 x(1)]);

% One of the asymptotically stable equilibria
theta_eq = 0;
x_eq = 0;
    
% Draw phase portrait
saddles = [x_eq, -theta_eq+period/2-2*period;...
    x_eq, -theta_eq+period/2-period;...
    x_eq, -theta_eq+period/2;...
    x_eq, -theta_eq+period/2+period;...
    x_eq, -theta_eq+period/2+2*period];
focuses = [x_eq, theta_eq-period;...
    x_eq, theta_eq-period;...
    x_eq, theta_eq;...
    x_eq, theta_eq+period;...
    x_eq, theta_eq+2*period];
    
plot(-focuses(:,2), -focuses(:,1), 'r.', 'MarkerSize', 20);
plot(focuses(:,2), focuses(:,1), 'k.', 'MarkerSize', 20);
plot(-saddles(:,2), -saddles(:,1), 'r.', 'MarkerSize', 20);
plot(saddles(:,2), saddles(:,1), 'k.', 'MarkerSize', 20);

% Jacobian matrix of the system
A = [0 1; 
    2/pi 2*a/pi];

% Calculating saddle eigenvectors V
[V, D] = eig(A);

% Custom simulation options
options = odeset('MaxStep', 0.001, 'RelTol', 2e-7, 'AbsTol', 2e-7);
draw_saddles_symmetric(pll_s,...
    [x_eq,period/2-theta_eq],  ...
    V, ...
    [-2*period,-period,0,period,2*period],...
    options);

% Plot adjustments
axis([-2*pi 2*pi -5 5])
xticks([-4*pi -3*pi -2*pi -pi 0 pi 2*pi, 4*pi])
xticklabels({'-4\pi','-3\pi','-2\pi','-\pi','0','\pi','2\pi','4\pi'})

xlabel('\theta_e');
ylabel('y');




\end{lstlisting}


\end{document}